\documentclass[a4paper,onecolumn,11pt,accepted=2021-10-14]{quantumarticle}

\pdfoutput=1
\usepackage[utf8]{inputenc}
\usepackage[english]{babel}
\usepackage[T1]{fontenc}
\usepackage{amsmath}
\usepackage{hyperref}

\usepackage{qcircuit}
\usepackage{braket}
\usepackage{amsfonts}
\usepackage{amssymb}
\usepackage{amsthm}

\newtheorem{definition}{Definition}
\newtheorem{theorem}[definition]{Theorem}
\newtheorem{proposition}[definition]{Proposition}
\newtheorem{lemma}[definition]{Lemma}
\newtheorem{corollary}[definition]{Corollary}

\newcommand{\eps}{\varepsilon}

\begin{document}

\title{Faster Coherent Quantum Algorithms for Phase, Energy, and Amplitude Estimation}

\author{Patrick Rall}
\affiliation{Quantum Information Center, University of Texas at Austin}


\maketitle

\begin{abstract}
    We consider performing phase estimation under the following conditions: we are given only one copy of the input state, the input state does not have to be an eigenstate of the unitary, and the state must not be measured. Most quantum estimation algorithms make assumptions that make them unsuitable for this `coherent' setting, leaving only the textbook approach. We present novel algorithms for phase, energy, and amplitude estimation that are both conceptually and computationally simpler than the textbook method, featuring both a smaller query complexity and ancilla footprint. They do not require a quantum Fourier transform, and they do not require a quantum sorting network to compute the median of several estimates. Instead, they use block-encoding techniques to compute the estimate one bit at a time, performing all amplification via singular value transformation. These improved subroutines accelerate the performance of quantum Metropolis sampling and quantum Bayesian inference.
\end{abstract}

Phase estimation is one of the most widely used quantum subroutines, because it grants quantum computers two unique capabilities. First, it yields a quadratic speedup in the accuracy of Monte Carlo estimates \cite{0005055,1504.06987,1807.06456}. Achieving `Heisenberg limited' accuracy scaling has numerous applications in physics, chemistry, machine learning, and finance \cite{2010.04173,2012.06681,2004.06832,2012.06283}. Second, it allows quantum computers to diagonalize unitaries in a certain restricted sense: if $U = \sum_{j} e^{2\pi i \lambda_j} \ket{\psi_j}\bra{\psi_j}$ then phase estimation performs the transformation
\begin{align}
 \sum_j \alpha_j \ket{0^n} \ket{\psi_j} \to \sum_j \alpha_j \ket{\lambda_j} \ket{\psi_j} \label{eqn:coherent}
\end{align}
where $\ket{\lambda_j}$ is an $n$-bit estimate of $\lambda_j$.
This access to spectral information enables quantum speedups for linear algebra \cite{0811.3171,1511.02306}, studying physical systems \cite{0911.3635,1011.1468, 1910.01659,2005.13434}, estimating partition functions \cite{1504.06987}, and performing Bayesian inference \cite{1907.09965,2009.11270}.

Phase estimation is also a very complicated algorithm. Textbook phase estimation \cite{nc} requires the Quantum Fourier Transform (QFT), a tool we commonly associate with the exponential speedup of Shor's algorithm \cite{9508027}. However, phase estimation only delivers a quadratic speedup for estimation and the exponential speedup for linear algebra can sidestep the QFT \cite{1511.02306,1806.01838}. When applied to energies, phase estimation also requires Hamiltonian simulation, a quantum subroutine that is an entire subject of study in its own right: it requires recent innovations to apply optimally in a black-box setting \cite{1501.01715,1606.02685,1610.06546,1707.05391,1806.01838} and optimal Hamiltonian simulation for specific systems is still being actively studied \cite{2012.09194,2012.09238}. Furthermore, phase estimation demands median amplification to guarantee an accurate answer. This can be challenging to implement coherently on a quantum computer \cite{0102078,0211174,1207.2307}, because it requires many ancillae and a quantum sorting network. The probability with which phase estimation gives the correct answer (sometimes referred to as the `Fejer kernel' \cite{2004.04889}) is not always high enough to be amplified, which must be dealt with either by rounding or adaptivity \cite{1610.09619}. The conceptual complexity of textbook phase estimation and the resulting computational overhead motivates a search for alternatives.

Fortunately, a lot of simplification is possible if we allow for `incoherent' algorithms, where incoherence manifests via a variety of assumptions. We could be allowed to measure the quantum state, either because we are given many copies of the input state, or because we have the ability to prepare it inexpensively. Alternatively, we can assume a type of adaptivity where quantum states wait patiently without decohering while a classical computer performs a computation on the side. This assumption sometimes lets us repair the input state after using it \cite{0506068,0911.3635,1711.11025}. But most importantly, incoherent phase estimation algorithms usually require that the input state is an eigenstate of the unitary or Hamiltonian in question. 

If enough of these assumptions hold, then `iterative phase estimation' \cite{9511026} removes an enormous amount of conceptual and computational overhead. The estimate can be extracted one bit at a time, thereby removing the QFT. The accuracy of each bit can be amplified individually via many classical samples, removing the need for the quantum sorting network. The awkward notion of an `$n$-bit estimate' can be removed and replaced with a traditional additive-error estimate \cite{1610.09619}. Iterative phase estimation has seen many refinements and has been applied to gate set tomography and ground state energy estimation \cite{0709.2996,1502.02677,2102.11340}. An incoherent iterative approach also permits direct simplification of amplitude estimation \cite{1908.10846, 1912.05559}.

However, for some applications of phase estimation maintaining coherence remains crucial. The original strategy for quantum matrix inversion \cite{0811.3171}, quantum Metropolis sampling \cite{0911.3635,1011.1468,1910.01659}, and a protocol for partition function estimation \cite{1504.06987}, thermal state preparation, and Bayesian inference \cite{1907.09965,2009.11270} all violate the assumptions above. The eigenvalues must be estimated while preserving the superposition, and there is no guarantee that the input state is an eigenstate. These applications of `coherent' phase estimation motivate the main question of this paper: does there exist a conceptually and computationally simpler algorithm for performing the transformation (\ref{eqn:coherent}) while remaining coherent? That is: the input state is not necessarily an eigenstate, we are given exactly one copy, and there is no adaptive interaction with a classical computer. 

In this paper we answer this question in the affirmative. We present simplified coherent algorithms for phase estimation, energy estimation, and amplitude estimation. None of these algorithms require a QFT, and they can be made arbitrarily close to the ideal transformation without a quantum sorting network to compute a median. These algorithms are about 14x to 10x faster than traditional phase estimation in terms of their query complexity, a performance metric that neglects the fact that they also require fewer ancilla qubits. 

However, we also observe that there are unavoidable barriers to estimation in superposition. Consider an estimation algorithm that outputs a superposition of two estimates $\hat\lambda^{(1)}_{j}$ and $\hat\lambda^{(2)}_j$, both of which could be pretty close to the true value $\lambda_j$:
\begin{align}
    \sum_j \alpha_j\ket{0^n}  \ket{\psi_j} \to \sum_j\alpha_j \left( \xi_j \ket{\lambda^{(1)}_j} + \zeta_j \ket{\lambda_j^{(2)}} \right) \ket{\psi_j} \label{eqn:nondeterministic}
\end{align}
However, it is a well known fact \cite{9701001} that uncomputation cannot work in such a situation (unless one of $\xi_j$ or $\zeta_j$ is $\approx 0$). Thus, any algorithm that actually makes use of the estimates must necessarily damage the input superposition and can no longer be considered coherent.

Even worse, we show that \emph{any unitary quantum algorithm} for estimation must perform a map in the form (\ref{eqn:nondeterministic}), and there always exist some values of $\lambda_j$ where the neither of the $\xi_j$ and $\zeta_j$ are $\approx 0$. Therefore, the only way to get an algorithm that performs the deterministic transformation (\ref{eqn:coherent}) is to assume certain values of $\lambda_j$ do not appear. We refer to this as a \textbf{rounding promise,} and show that if the rounding promise holds, then our algorithms perform the map (\ref{eqn:coherent}). However, we also construct our algorithms in such a way that they give reasonable non-deterministic estimates as in (\ref{eqn:nondeterministic}) even when no rounding promise holds.

In the following we give a brief outline of the method we use to construct the algorithms. They strongly resemble iterative phase estimation \cite{9511026}, which works roughly like this: the estimate is computed one bit at a time, starting with the least significant bit. A `Hadamard-test' computes each bit with a decent success probability, which is then amplified to a high probability by taking the majority vote of many samples. This process could be made coherent naively by computing each sample into an ancilla, but this requires so many ancillae that any performance benefit over the QFT- and median-based approach is lost. The key idea is to amplify without using any new ancillae.

To manipulate the probabilities of the Hadamard-test, we use `block-encodings'. Block-encodings permit quantum computers to manipulate non-unitary matrices. In our case, the matrices' eigenvalues encode our probabilities. A unitary matrix $U_A$ is a block-encoding of $A$ if $A$ is in the top-left corner:
\begin{align}
    U_A = \begin{bmatrix} A & \cdot\hspace{1mm} \\ \cdot & \cdot \end{bmatrix}
\end{align}
Block-encoded matrices can be manipulated in two ways. First, \emph{linear combinations of unitaries} \cite{1501.01715, 1511.02306} allow us to build block-encodings $\alpha A  + \beta B$ given block-encodings of $A$ and $B$. Linear combinations of unitaries allow us to make block-encodings of Hamiltonians presented as a sum of local terms, covering most practical applications. Second, \emph{singular value transformation} \cite{1610.06546, 1606.02685, 1806.01838} lets us apply certain polynomials $p(x)$ to the singular values of a block-encoded matrix $A$, using $\text{deg}(p)$ many queries to controlled-$U_A$ and its inverse. 

Together, these two techniques permit the unification of many quantum algorithms into a single framework. For an accessible introduction to these methods, along with a comprehensive review of the most modern versions of these algorithms we refer to \cite{2105.02859}. This work also presents the independent discovery of an algorithm very similar to our improved method of phase estimation, which is also sketched in \cite{Chuang20}. 

To obtain the $k$'th bit ($0 \leq k < n-1$), iterative phase estimation performs a Hadamard-test on $U^{2^{n-k-1}}$, where $U = \sum_j e^{2\pi i \lambda_j} \ket{\psi_j}\bra{\psi_j}$. The outcome of the test is a coin toss that is heads with probability $\cos^2( 2^{n-k-1} \pi \lambda_i) $. We observe that the quantum circuit for the Hadamard-test resembles a linear combination of unitaries. Therefore, we can construct a block-encoding of a matrix whose eigenvalues encode the Hadamard-test probability for each eigenvector of $U$.

Iterative phase estimation then proceeds to perform the Hadamard-test several times and takes the majority vote. We observe that the probability that the majority vote is 1 is a polynomial in the Hadamard-test probability $p$:
\begin{align}
    \text{Pr}[ \text{majority vote is 1}  ] = \sum_{k = \lceil M/2\rceil}^{M} \binom{M}{k} p^k (1-p)^{M-k}
\end{align}
We can therefore apply such an `amplifying polynomial' \cite{0902.3757} directly to the block-encoding using singular value transformation, which requires no new ancillae. In fact, using techniques from \cite{1707.05391}, we can construct a polynomial that performs the same task with much smaller degree.

Now we have amplified the eigenvalues of the block-encoding to be close to 0 or 1, so we have a projector. To extract the 0/1-eigenvalue into a qubit we use the following novel technical tool: 
\\[2mm]
\noindent \textbf{Theorem.} \textit{ \textbf{Block-measurement.} Say we have an approximate block-encoding of a projector $\Pi$. Then there exists a quantum channel that approximately implements the map:
    \begin{align}
        \ket{0}\otimes \ket{\psi} \to   \ket{1} \otimes \Pi\ket{\psi} + \ket{0} \otimes (I-\Pi)\ket{\psi}
    \end{align}}
The channel is based on uncomputation \cite{9701001}, but the error analysis also features an interesting trick to deal with the case where the uncomputation fails. This theorem may find applications elsewhere.

Repeating the above procedure for each bit while carefully adjusting the phases at every step yields an algorithm that performs coherent phase estimation with no ancillae required. To estimate energies, we can construct a block-encoding with eigenvalues $\cos^2(2^{n-k-1}\pi \lambda_j)$ directly using the Jacobi-Anger expansion \cite{1806.01838} rather than going through Hamiltonian simulation, which boosts the performance further, although now the algorithm does require ancillae. Finally, to perform coherent amplitude estimation, we construct a block-encoding of a 1x1 matrix containing the amplitude to be estimated and then invoke energy estimation.

    A key research question of this paper is: Do these algorithms perform better than traditional phase estimation in practice? An asymptotic analysis is not sufficient to answer this question. Instead, we carefully bound the query complexity and failure probability of all algorithms involved using the diamond norm and carefully select constants to maximize the performance. Subsequently, we perform a numerical analysis of the query complexity. We find that we improve the query complexity of phase estimation by a factor of about 14x, and the query complexity of energy estimation is improved by about 10x. Our amplitude estimation algorithm inherits the speedup of the energy estimation algorithm.

    This paper is structured as follows. In section~\ref{sec:prelims} we carefully define the problem of coherent estimation and establish the notion of a rounding promise. Then we analyze the textbook method. In section~\ref{sec:coherent} we present our novel algorithms for phase and energy estimation. In section~\ref{sec:performance} we discuss a numerical analysis of the query complexities of the algorithms. Then, in section~\ref{sec:blockmeas} give a proof of the block-measurement theorem above and then show how to use energy estimation to perform non-destructive amplitude estimation in section~\ref{sec:ampest}.

\section{Preliminaries} \label{sec:prelims}

In this section we formally define the estimation tasks we want to solve (Definition~\ref{def:estimator}). This requires setting up the notion of a `rounding promise' (Definition~\ref{def:roundingpromise}) which highlights an inherent complication with coherent estimation that is not present in the incoherent case. We argue that this complication is unavoidable, and should be taken into account when coherent estimation is performed in practice. Next, we set up some simple technical tools for dealing with the error analysis of uncomputation (Lemmas~\ref{lemma:garbageremove},\ref{lemma:spectraltodiamond}). These will be used several times in this paper. Finally, we give a precise description and analysis of `textbook' phase estimation (Proposition~\ref{prop:phaseestimation}). Thus, this section clearly defines the problem and the previous state-of-the-art which we improve.

Our paper presents algorithms for phase, energy, and amplitude estimation. Amplitude estimation follows as a relativity simple corollary of energy estimation so we will only be talking about the prior two until Section~\ref{sec:blockmeas}. To talk about both phases and energies at once, we standardize input unitaries and Hamiltonians into the following form:
\begin{align}
    U = \sum_{j} e^{2\pi i \lambda_j} \ket{\psi_j}\bra{\psi_j} \hspace{2cm} H = \sum_j \lambda_j \ket{\psi_j}\bra{\psi_j}
\end{align}
We refer to the $\{\lambda_j\}$ as the `eigenvalues', and assume they live in the range $[0,1)$. While any unitary can be put into this form, the form places the constraint $0 \preceq H \prec I$ onto the Hamiltonian.

Our goal is to compute $\ket{\lambda_j}$ an `$n$-bit estimate of $\lambda_j$'. In this paper, we take this to be an $n$-qubit register containing a binary encoding of the number $\text{floor}(2^n \lambda_j)$. Since $\lambda_j \in [0,1) $ we are guaranteed that $\text{floor}(2^n \lambda_j)$ is an integer $\in \{0, ..., 2^n -1\}$ and thus has an $n$-bit encoding.

Consider phase estimation using a quantum circuit composed of elementary unitaries and controlled-$e^{2\pi i\lambda_j}$. Following an argument related to the polynomial method, we see that the resulting state must be of the form:

\begin{align}
    \sum_{x\in\{0,1\}^n} \sum_{y}  \alpha_{x,y}(e^{2\pi i\lambda_j}) \ket{x} \ket{ \text{garbage}_{x,y} } 
\end{align}
where $\alpha_{x,y}(e^{2\pi i\lambda_j})$ is some polynomial of $e^{2\pi i\lambda_j}$. We would like $\ket{x}$ to encode $\text{floor}(2^n \lambda_j)$, meaning that $\alpha_{x,y}(e^{2\pi i\lambda_j}) = 1$ if $x = \text{floor}(2^n \lambda_j)$ and $\alpha_{x,y}(e^{2\pi i\lambda_j}) = 0$ otherwise. This is impossible: $\alpha_{x,y}(e^{2\pi i\lambda_j})$ is a continuous function of $\lambda_j$, but the desired amplitude indicating $x = \text{floor}(2^n \lambda_j)$ is discontinuous. For energy estimation a similar argument applies, just that the amplitudes are of the form $\alpha_{x,y}(\lambda_j)$. This argument even holds in the approximate case when we demand that $\alpha_{x,y} \leq \delta$ or $\geq 1-\delta$ for some small $\delta$ - the discontinuity is present regardless. 

To some extent, this issue stems from the awkwardness of the notion of an `$n$-bit estimate' since it requires rounding or flooring, a discontinuous operation, when only continuous manipulation of amplitudes is possible. A more comfortable notion is that of an additive-error estimate, used by \cite{1610.09619,1603.08675} in their estimation algorithms.

However, the primary application of our algorithms is a situation where this simplification is not possible: Szegedy walks based on Hamiltonian eigenspaces \cite{1011.1468, 1910.01659,2005.13434,2107.07365}. The original quantum Metropolis algorithm \cite{0911.3635} implements a random walk over Hamiltonian eigenspaces, where the superposition is measured at every step and it is demonstrated that an additive error estimate is sufficient.  But in order to harness a quadratic speedup due to quantum walks \cite{Sze04}, the measurement of energies must be made completely coherent, which is not achieved by an additive-error estimate unless the error is smaller than the gap between any eigenvalues. Therefore, we retain the notion of an `$n$-bit estimate' in this paper. We could always fall back to an estimate with additive error $\eps/2$ by computing $n = \text{ceil}(\log_2(\eps^{-1}))$ bits of accuracy.  

However, the above argument still applies for additive-error estimation: the amplitude of any given estimate $\ket{\hat\lambda}$ is a continuous function of the eigenvalue $\lambda_j$. This means that the amplitude cannot be 0 or 1 everywhere, there must exist points where it crosses intermediate values in order to interpolate in between the two. In both the `$n$-bit estimate' case and the additive-error case, this causes a problem for coherent quantum algorithms since the estimate cannot always be uncomputed. For some eigenvalues $\lambda_j$, the algorithm will yield an output state of the form $\alpha \ket{\hat\lambda} + \beta \ket{\hat\lambda'}$. Even if $ \hat\lambda,\hat\lambda'$ are good estimates of $\lambda_j$, the uncompute trick \cite{9701001} cannot be used for such an output state. Thus, the damage to the input superposition over $\ket{\lambda_j}$ is irreparable. 

The only way to deal with this issue is to assume that the $\lambda_j$ do not take certain values. Then the amplitudes can interpolate between 0 and 1 at those points. Observe that the discontinuities in the function $\text{floor}(2^n \lambda_j)$ occur at multiples of $1/2^n$. A `rounding promise' simply disallows eigenvalues near these regions.


\begin{definition} \label{def:roundingpromise} Let $n\in \mathbb{Z}^+$ and $\alpha \in (0,1)$. A hermitian matrix $H$ satisfies an $(n,\alpha)$-\textbf{rounding promise} if it has an eigendecomposition $H = \sum_j \lambda_j \ket{\psi_j}\bra{\psi_j}$, all the eigenvalues $\lambda_j$ satisfy $0 \leq \lambda_j < 1$, and for all $x \in \{0,...,2^n\}$:
    \begin{align}
        \lambda_j \not\in \left[ \frac{x}{2^n} , \frac{x}{2^n} + \frac{\alpha}{2^n}  \right ]
    \end{align}

    Similarly, a unitary matrix $U$ satisfies an $(n,\alpha)$-rounding promise if it can be written as $U = \sum_j e^{2\pi i \lambda_j}\ket{\psi_j}\bra{\psi_j}$ and the phases $\lambda_j$ satisfy the same assumptions.
\end{definition}


We have $\lambda_j \in [0,1)$, and we have disallowed $\lambda_j$ from certain sub-intervals of $[0,1)$. The definition above is chosen such that the total length of these disallowed subintervals is $\alpha$, regardless of the value of $n$. We have essentially cut out an $\alpha$-fraction of the allowed eigenvalues. 

Guaranteeing a rounding promise demands an enormous amount of knowledge about the eigenvalues. We do not expect many Hamiltonians or unitaries in practice to actually provably satisfy such a promise. However, we expect that the estimation algorithms discussed in this paper will still perform pretty well even if they do not satisfy such a promise. One can increase the chances of success by setting $\alpha$ to be very small, so that the vast majority of the eigenvalues do not fall into a disallowed region. Then, if the input state is close to a uniform distribution over the eigenstates, then one can be fairly confident that only an $\alpha$ fraction of the eigenvalues in the support will be disallowed. The need for a rounding promise can be also sidestepped entirely by avoiding the use of energy estimation as a subroutine and approaching the desired problem directly. This has been accomplished for ground state finding \cite{2002.12508}.

The rounding promise is not just a requirement of our work in particular - it is a requirement for any coherent phase estimation protocol. The fact that coherent phase estimation does not work for certain phases is largely disregarded in the literature: for example, works like \cite{1603.08675} neglect this issue entirely. However, there are some works that have observed this problem and attempt to mitigate it. Such methods are called `consistent phase estimation' \cite{tashma13} since the error of their estimate is supposedly independent of the phase being estimated, thus allowing amplification to make the amplitudes always close to 0 or 1. We claim that all of these attempts fail, and furthermore that achieving coherent phase estimation without some kind of promise is impossible in principle. This is due to the polynomial method argument above: any coherent quantum algorithm's output state's amplitudes must be a continuous function of the phase, and continuous functions that are sometimes $\approx 0$ and sometimes $\approx 1$ must somewhere have an intermediate value. Recall that uncomputation only works when the amplitudes are close to 0 or 1. The error depends on if the phase is close to this transition point or not, so the error must depend on the phase. This issue was not taken into account by \cite{1010.4458}, \cite{1704.04992}, and \cite{1812.03584}, since all of these either implicitly or explicitly state that there is an algorithm that approximately performs $\ket{\psi_i}\ket{0^n}\to\ket{\psi_i}\ket{\lambda_i}$ in superposition while $\lambda_i$ is a deterministic computational basis state\footnote{\cite{1010.4458} makes use of such a map in Algorithm~4.  \cite{1704.04992} states this as Theorem~II.2.  \cite{1812.03584} sketches but does not analyze a protocol for this after Claim~4.5.}. A more promising approach is detailed in \cite{tashma13}, which, crudely speaking, shifts the transition points by a classically chosen random amount.  Now whether or not a phase is close to a transition point is independent of the phase itself, making amplification possible. \cite{1010.4458} describes a similar idea, calling it `unique-answer' eigenvalue estimation. However, we claim that this only works for a single phase. If we consider, for example, a unitary whose phases are uniformly distributed in $[0,1)$ at a sufficiently high density, then there will be a phase near a transition point for any choice of random shift. The only way to avoid the rounding promise is to sacrifice coherence and measure the output state. 

All of the algorithms in this paper, including textbook phase estimation, achieve an asymptotic runtime of $O(2^n \alpha^{-1} \log(\delta^{-1}))$, where $\delta$ is the error in diamond norm. Before we move on to the formal definition of the estimation task, we informally argue that the $\alpha^{-1}$ dependence is optimal, via a reduction to approximate counting. We are given $N$ items, $K$ of which are marked. Following the standard method for approximate counting \cite{0005055}, we construct a Grover unitary whose phases $\lambda_j$ encode $\arcsin(\sqrt{K/N})$. Given a $(1,\alpha)$-rounding promise, computing $\text{floor}(2^1 \lambda_j )$ amounts to deciding if $\lambda_j \leq \frac{1-\alpha/2}{2}$ or $\lambda_j \geq \frac{1+\alpha/2}{2}$ given that one of these is the case.  By shifting the $\lambda_j$ around appropriately we can thus decide if $K \geq (1/2 +  C\alpha )N $ or $K \leq (1/2- C \alpha)N$, for some constant $C$ obtained by linearising $\arcsin(\sqrt{K/N})$. We have achieved approximate counting with a promise gap $\sim\alpha$. Thus the $\Omega(\alpha^{-1})$ lower bound on approximate counting \cite{9804066} implies our runtime must be $\Omega(\alpha^{-1})$.

Equipped with the notion of a rounding promise, we can define our estimation tasks. Many algorithms in this paper produce some kind of garbage, which can be dealt with the uncompute trick \cite{9701001}. Rather than repeat the analysis of uncomputation in every single proof, we present a modular framework where we can deal with uncomputation separately. Furthermore, some applications may require computing some function of the final estimate, resulting in more garbage which also needs to be uncomputed. Rather than baking the uncomputation into each algorithm, it is thus more efficient to leave the decision of when to uncompute to the user of the subroutine. 


\begin{definition}\label{def:estimator} A \textbf{phase estimator} is a protocol that, given some $n \in \mathbb{Z}^+$ and $\alpha \in (0,1)$, and any error target $\delta > 0$, produces a quantum circuit involving controlled $U$ and $U^\dagger$ calls to some unitary $U$. If $U = \sum_j e^{2\pi i \lambda_j}\ket{\psi_j}\bra{\psi_j}$ satisfies an $(n,\alpha)$-rounding promise then this circuit implements a quantum channel that is $\delta$-close in diamond norm to the map:
    \begin{align}
        \ket{0^n} \ket{\psi_j} \to \ket{\text{floor}(\lambda_j 2^n)}\ket{\psi_j}
    \end{align}

    Similarly, an \textbf{energy estimator} is such a protocol that instead involves such calls to $U_H$ which is a block-encoding (see Definition~\ref{def:blockencoding}) of a Hamiltonian $H =  \sum_j \lambda_j \ket{\psi_j}\bra{\psi_j}$ that satisfies an $(n,\alpha)$-rounding promise.

    The \textbf{query complexity of an estimator} is the number of calls to $U$ or $U_H$ in the resulting circuit, as a function of $n,\alpha,\delta$.

    An estimator is said to `\textbf{have garbage}' or `\textbf{have $m$ qubits of garbage}' if it is instead close to a map that produces some $j$-dependent $m$-qubit garbage state in another register:
    \begin{align}
        \ket{0^n}\ket{0...0} \ket{\psi_j} \to \ket{\text{floor}(\lambda_j 2^n)} \ket{\text{garbage}_j} \ket{\psi_j}
    \end{align}
    (Note that the quantum circuit can allocate and discard ancillae, but that does not count as garbage.)

    An estimator is said to `\textbf{have phases}' if the map introduces a $j$-dependent phase $\varphi_j$:
    \begin{align}
        \ket{0^n} \ket{\psi_j} \to e^{i\varphi_j} \ket{\text{floor}(\lambda_j 2^n)}\ket{\psi_j}
    \end{align}
\end{definition}

If an estimator is both `with phases' and `with garbage', then we can just absorb the $e^{i\varphi_j}$ into $\ket{\text{garbage}_j}$ so the `with phases' is technically redundant. However, in our framework it makes more sense to treat phases and garbage independently since some of our algorithms are just `with phases'.


The reason we measure errors in diamond norm has to do with uncomputation of approximate computations with garbage. Consider for example a transformation $V$ acting on an answer register and a garbage register:
\begin{align}
    V \ket{0} \ket{0...0} =   \sqrt{1-\eps} \ket{0}\ket{\text{garbage}_0} + \sqrt{\eps} \ket{1} \ket{\text{garbage}_1} 
\end{align}
for some small nonzero $\eps$. We copy the answer register into the output register:
\begin{align}
    \to   \sqrt{1-\eps} \ket{0} \otimes  \ket{0}\ket{\text{garbage}_0}  + \sqrt{\eps} \ket{1}  \otimes \ket{1} \ket{\text{garbage}_1}
\end{align}
and then we project the answer and garbage registers onto $V\ket{0}\ket{0...0}$:
\begin{align}
    \to   \sqrt{1-\eps} \ket{0} \cdot \sqrt{1-\eps} +\sqrt{\eps} \ket{1} \cdot \sqrt{\eps} 
\end{align}
The resulting state $(1-\eps)\ket{0} + \eps\ket{1}$ is not normalized, meaning that the projection $V\ket{0}\ket{0...0}$ succeeds with some probability $< 1$. Thus the uncomputed registers are not always returned to the $\ket{0}\ket{0...0}$ state. 

At this point our options are either to postselect these registers to $\ket{0}\ket{0...0}$ or to discard them. Postselection improves the accuracy, but also implies that the algorithms do not always succeed. Since many applications of coherent energy estimation demand repeating this operation many times, it is important that the algorithm always succeeds. Thus, we need to discard qubits, so we need to talk about quantum channels. Therefore, the diamond norm is the appropriate choice for an error metric. Recall that the diamond norm is defined in terms of the trace norm \cite{9806029}:
\begin{align}
    \left| \Lambda \right|_\diamond := \sup_\rho \left| (\Lambda \otimes \mathcal{I}) (\rho) \right|_1
\end{align}

where $\mathcal{I}$ is the identity channel for some Hilbert space of higher dimension than $\Lambda$. The following analysis shows how to remove phases and garbage from an estimator, even in the approximate case.


\begin{lemma} \label{lemma:garbageremove} \textbf{Getting rid of phases and garbage.} Given a phase/energy estimator with phases and/or garbage with query complexity $Q(n,\alpha,\delta)$ that is unitary, we can construct a phase/energy estimator without phases and without garbage with query complexity $2Q(n,\alpha, \delta/2)$. 
\end{lemma}
\begin{proof} Without loss of generality, we assume we are given a quantum channel $\Lambda$ that implements something close in diamond norm to the map:
    \begin{align}
        \ket{0^n} \ket{0...0}\ket{\psi_j} \to e^{i\varphi_j} \ket{\text{floor}(\lambda_j 2^n)} \ket{\text{garbage}_j} \ket{\psi_j} \label{eqn:idealuncompute}
    \end{align}
    If the channel is actually without phases then we can just set $\varphi_j =0$, and if it is actually without garbage then the following calculation will proceed without problems. Our strategy is just to use the uncompute trick, and then to discard the uncomputed ancillae. This requires us to implement $\Lambda^{-1}$, so we required that $\Lambda$ be implementable by a unitary.
\begin{align}
       \hspace{5mm} \begin{array}{c}\Qcircuit @C=1em @R=1em {
                    \lstick{\ket{0^n}}   & \qw & \targ & \qw  & \qw\\
                \lstick{\ket{0^n}} & \multigate{2}{\Lambda} & \ctrl{-1} & \multigate{2}{\Lambda^{-1}} &  \measuretab{\text{discard}}\\
           \lstick{\ket{0...0}}     & \ghost{\Lambda} & \qw & \ghost{\Lambda^{-1}} & \measuretab{\text{discard}}\\
           \lstick{\ket{\psi_j}}    & \ghost{\Lambda} & \qw & \ghost{\Lambda^{-1}} & \qw 
       }\end{array} \label{eqn:uncomputechannel}
\end{align}

    First, we consider the ideal case when $\Lambda$ implements the map (\ref{eqn:idealuncompute}) exactly. If we use $\ket{\psi_j}$ as input and we stop the circuit before the discards, we produce a state $\ket{\text{ideal}_j}$:
\begin{align}
    \ket{0^n} \ket{0^n} \ket{0...0}\ket{\psi_j} &\to e^{i\varphi_j} \ket{0^n}\ket{\text{floor}(\lambda_j 2^n)} \ket{\text{garbage}_j} \ket{\psi_j}\\
     &\to e^{i\varphi_j} \ket{\text{floor}(\lambda_j 2^n)}\ket{\text{floor}(\lambda_j 2^n)} \ket{\text{garbage}_j} \ket{\psi_j}\\
     &\to  \ket{\text{floor}(\lambda_j 2^n)}\ket{0^n} \ket{0...0} \ket{\psi_j} =: \ket{\text{ideal}_j}
\end{align}

    Let $\rho^\text{ideal}_{i,j} := \ket{\text{ideal}_i}\bra{\text{ideal}_j}$, so that we can write down the ideal channel:
    \begin{align}
        \Gamma_\text{ideal}(\sigma) := \sum_{i,j} \bra{\psi_i}\sigma\ket{\psi_j} \cdot \rho^\text{ideal}_{i,j}
    \end{align}

    If we apply $\Lambda$ with error $\delta/2$ instead we obtain the actual channel $\Gamma$ such that:
    \begin{align}
        \sup_\sigma  \left| (\Gamma \otimes \mathcal{I})(\sigma) - (\Gamma_\text{ideal} \otimes \mathcal{I})(\sigma)  \right|_1 \leq \delta
\end{align}
    If $A$ refers to the subsystems that start (and approximately end) in the states $\ket{0^n}\ket{0...0}$ then the final output state satisfies:
    \begin{align}
        &\sup_\sigma   \left|\text{Tr}_A\left((\Gamma \otimes \mathcal{I})(\sigma)\right) - \text{Tr}_A\left((\Gamma_\text{ideal} \otimes \mathcal{I})(\sigma) \right) \right|_1 \\
        \leq  &\sup_\sigma \left|\text{Tr}_A\left((\Gamma \otimes \mathcal{I})(\sigma) - (\Gamma_\text{ideal} \otimes \mathcal{I})(\sigma) \right) \right|_1 \\
        \leq  &\sup_\sigma \left|(\Gamma \otimes \mathcal{I})(\sigma) - (\Gamma_\text{ideal} \otimes \mathcal{I})(\sigma)  \right|_1 \leq \delta
\end{align}
    where we have used the fact that for all $\rho$ and subsystems $A$ we have $|\text{Tr}_A(\rho)|_1 \leq |\rho|_1$. Upon plugging in $\sigma = \ket{\psi_j}\bra{\psi_j}$ we can see that tracing out the middle register after applying $\Gamma_\text{ideal}$ implements the map
    \begin{align}
        \ket{0^n}\ket{\psi_j} \to \ket{\text{floor}(\lambda_j2^n)}\ket{\psi_j}
    \end{align}
    so therefore the circuit in (\ref{eqn:uncomputechannel}) is an estimator without phases and garbage as desired. 

    The proof is complete, up to the fact that $|\text{Tr}_A(\rho)|_1 \leq |\rho|_1$ for all $\rho$ and subsystems $A$. Write $\rho$ in terms of its eigendecomposition $\rho = \sum_i \lambda_i\ket{\phi_i}\bra{\phi_i}$, and let $\rho_i := \text{Tr}_A(\ket{\phi_i}\bra{\phi_i})$:
    \begin{align}
        \left|\text{Tr}_A(\rho)\right|_1 &=  \left|\text{Tr}_A\left( \sum_i \lambda_i \ket{\phi_i}\bra{\phi_i}\right )\right|_1=\left| \sum_i \lambda_i \rho_i \right|_1  \leq \sum_i |\lambda_i| \cdot |\rho_i|_1 = |\rho|_1
    \end{align}

\end{proof}


This establishes that uncomputation works in the approximate case as expected. While we are reasoning about the diamond norm, we also present the following technical tool which will come in handy several times. In particular, we will need it for our analysis of textbook phase estimation.

\begin{lemma} \label{lemma:spectraltodiamond} \textbf{Diamond norm from spectral norm}. Say $U,V$ are unitary matrices satisfying $|U-V| \leq \delta$. Then the channels $\Gamma_U(\rho):= U\rho U^\dagger$ and $\Gamma_V(\rho):= V\rho V^\dagger$ satisfy $|\Gamma_U - \Gamma_V|_\diamond \leq  2\delta$.

\end{lemma}
\begin{proof}
    We have that:
    \begin{align}
        \left| \Gamma_U  - \Gamma_V \right|_\diamond &:= \sup_\rho \left| (\Gamma_U \otimes \mathcal{I})(\rho)  -  (\Gamma_V \otimes \mathcal{I})(\rho)  \right|_1\\
        &= \sup_\rho \left| (U\otimes I)\rho(U \otimes I)^\dagger  - (V\otimes I)\rho(V \otimes I)^\dagger  \right|_1
    \end{align}
where $\mathcal{I}$ was the identity channel on some subsystem of dimension larger than that of $U,V$, and $|M|_1$ is the sum of the magnitudes of the singular values of $M$.

Let $\bar U = U \otimes I$ and $\bar V = V \otimes I$. Then:
    \begin{align}
        \left| \bar U - \bar V \right| = \left| U \otimes I - V \otimes I   \right| =  \left| (U-V) \otimes I  \right| =   \left| U-V\right| \leq \delta
    \end{align}
If we let $\bar E := \bar U - \bar V$ we can proceed with the upper bound:
    \begin{align}
        \left| \Gamma_U  - \Gamma_V \right|_\diamond &= \sup_\rho \left| \bar U\rho\bar U^\dagger  -\bar V\rho \bar V^\dagger  \right|_1\\
        &= \sup_\rho \left| \bar U\rho\bar U^\dagger  -\bar V\rho \bar V^\dagger + \left( \bar U \rho V^\dagger - U\rho V^\dagger\right)  \right|_1\\
        &= \sup_\rho \left|\left(\bar U\rho\bar U^\dagger  - U\rho V^\dagger\right) - \left(\bar V\rho \bar V^\dagger -  \bar U \rho V^\dagger \right)  \right|_1\\
        &= \sup_\rho \left|\bar U \rho (\bar U - \bar V)^\dagger + (\bar U - \bar V)\rho \bar V^\dagger  \right|_1\\
        &\leq \delta + \delta 
    \end{align}

\end{proof}


Now we have all the tools required to analyze the textbook algorithm \cite{nc}. This algorithm combines several applications of controlled-$U$ with an inverse QFT to obtain the correct estimate with a decent probability. One can then improve the success probability via median amplification: if a single estimate is correct with probability $\geq \frac{1}{2} + \eta$ for some $\eta$ then the median of $\lceil \ln(\delta^{-1})/ 2\eta^2 \rceil$ estimates is incorrect with probability $\leq \delta$.

However, median amplification alone is not sufficient to accomplish phase estimation as we have defined it in Definition~\ref{def:roundingpromise}. This is because any $\lambda_i$ that is $\approx 10\% \cdot 2^{-n-1}$ close to a multiple of $1/2^{n}$, the probability of correctly obtaining $\text{floor}(2^n \lambda_j)$ is actually \emph{less} than $1/2$! This means that no matter how much median amplification is performed, this approach cannot achieve $\alpha$ below $\approx 10\%$. (For reference, a lower bound $\gamma$ on the probability is plotted in Figure~\ref{fig:alpha_eta}, although reading this figure demands some notation from the proof of the following proposition. We see that when $|\lambda_j^{(x)} - 1/2| \lesssim 10\% / 2 $ the probability of obtaining $\text{floor}(2^n \lambda_j)$ is less than $1/2$.)

Furthermore, even if the signal obtained from the QFT could obtain arbitrarily small $\alpha$, it would not achieve the desired $\alpha^{-1}$ scaling. This is because the amplification gap $\eta$ scales linearly with $\alpha$, and median amplification scales as $\eta^{-2}$ due to the Chernoff-Hoeffding theorem. Therefore, we would obtain a scaling of $\alpha^{-2}$.

Fortunately, achieving $\sim\alpha^{-1}$ is possible, using a different method for reducing $\alpha$: we perform estimation for $r$ additional bits which are then ignored. This makes use of the fact that $\alpha$ is independent of the number of estimated bits, but the gap away from the multiples of $1/2^{n}$, which is $\alpha/2^{n+1}$, is not. Every time we increase $n$, the width of each disallowed interval is chopped in half, but to compensate the number of disallowed intervals is doubled. If we simply round away some of the bits, then we do not care about the additional disallowed regions introduced. If we estimate and round $r$ additional bits, we suppress $\alpha$ by a factor of $2^{-r}$ while multiplying our runtime by a factor of $2^r$ - so the scaling is $\sim\alpha^{-1}$.

We still need median amplification, though. In order to use the above strategy we need rounding to be successful: if an eigenvalue $\lambda_j$ falls between two bins $\text{floor}(2^n\lambda_j)$ and $\text{floor}(2^n\lambda_j)+1$, then it must be guaranteed to be rounded to one of those bins with failure probability at most $\delta$. Before amplification, the probability that rounding succeeds is $\geq 8/\pi^2$, a constant gap above $1/2$ corresponding to $\alpha = 1/2$. We cannot guarantee a success probability higher than $8/\pi^2$ using the rounding trick alone, and thus need median amplification to finish the job.

Below, we split our analysis into two regimes: the $\alpha \leq 1/2$ regime and the $\alpha > 1/2$ regime. When $\alpha > 1/2$ then we do not really need the rounding trick described above, and we can achieve this accuracy with median amplification alone. Otherwise when $\alpha \leq 1/2$ we use median amplification to get to the point where rounding is likely to succeed, and then estimate $r$ additional bits such that $\alpha 2^r \geq 1/2$.


\begin{proposition}\label{prop:phaseestimation} \textbf{Standard phase estimation.} There exists an phase estimator with phases and garbage. Consider a unitary satisfying an $(n,\alpha)$-rounding promise. Let:
    \begin{align}
        \gamma(x) := \frac{sin^2(\pi x)}{\pi^2 x^2}\hspace{1.5cm} \eta_0 := \frac{8}{\pi^2} - \frac{1}{2}\hspace{1.5cm} \delta_\text{med} := \frac{\delta^2}{6.25}
    \end{align}
If $\alpha \leq 1/2$,let $  r := \left\lceil \log_2 \left(\frac{1}{2\alpha}\right)  \right\rceil$.
Then the phase estimator has query complexity:
    \begin{align}
        (2^{n+r} -1) \cdot \left\lceil \frac{\ln(\delta^{-1}_\text{med})}{2 \eta_0^2} \right\rceil
    \end{align}
    and has $(n+r)\left\lceil \frac{\ln(\delta^{-1}_\text{med})}{2\eta_0^2}\right\rceil$ qubits of garbage. 

Otherwise, if $\alpha > 1/2$, let $\eta := \gamma\left( \frac{1-\alpha}{2} \right) - \frac{1}{2}$. Then the phase estimator has query complexity:
    \begin{align}
        (2^{n} -1) \cdot \left\lceil \frac{\ln(\delta^{-1}_\text{med})}{2 \eta^2} \right\rceil
    \end{align}
    and has $n\left\lceil \frac{\ln(\delta^{-1}_\text{med})}{2\eta^2}\right\rceil$ qubits of garbage. 
\end{proposition}

\begin{proof} We begin with the $\alpha > 1/2$ case which is simpler, and then extend to the $\alpha \leq 1/2$ case. The following is a review of phase estimation, which has been modified to estimate $\text{floor}(2^n \lambda_j)$ rather than $\text{round}\left( 2^n \lambda_j \right)$:

 \begin{enumerate}
        \item Prepare a uniform superposition over times: $\ket{+^{n}}\ket{\psi_j} =\frac{1}{\sqrt{2^n}} \sum_{t=0}^{2^n-1} \ket{t}\ket{\psi_j}$.
        \item Apply $U$ to the $\ket{\psi_j}$ register $t$ times, along with an additional phase shift of $ \frac{-2\pi t(1-\alpha)}{2^{n+1}}$ to account for flooring:
            \begin{align}
                \to  \frac{1}{\sqrt{2^n}}   \sum_{t=0}^{2^n-1}  \ket{t} \otimes U^t e^{-\frac{2\pi i t(1-\alpha)}{2^{n+1}}} \ket{\psi_j} =   \frac{1}{\sqrt{2^n}}  \sum_{t=0}^{2^n-1} e^{2\pi i t  \left(\lambda_j - \frac{1-\alpha}{2^{n+1}} \right) } \ket{t}\ket{\psi_j} 
            \end{align}
        \item Apply an inverse quantum Fourier transform to the $\ket{t}$ register:
            \begin{align}
                &\to \frac{1}{2^n}\sum_{t=0}^{2^n-1} \sum_{x=0}^{2^n-1} e^{2\pi i t \left(\lambda_j - \frac{1-\alpha}{2^{n+1}} \right)} e^{- \frac{2\pi i}{2^n} tx} \ket{x}\ket{\psi_j}  \\ 
                &=  \sum_{x=0}^{2^n-1} \left[ \frac{1}{2^n}\sum_{t=0}^{2^n-1}e^{  2i \pi t \left(\lambda_j - \frac{x}{2^n} - \frac{1-\alpha}{2^{n+1}} \right)}\right] \ket{x}\ket{\psi_j} \\
                &=  \sum_{x=0}^{2^n-1} \beta\left(\lambda^{(x)}_j \right) \ket{x}\ket{\psi_j} 
            \end{align}
            where in the final line we have defined:
            \begin{align}
                \lambda^{(x)}_j &:= 2^n \lambda_j  - x  - \frac{1-\alpha}{2} \\
                \beta(\lambda_j^{(x)}) &:= \frac{1}{2^n}\sum_{t=0}^{2^n-1}e^{  2 \pi i t \lambda^{(x)}_j/2^n }.
            \end{align}
        \item Let $\eta$ and $\delta_\text{med}$ be as in the theorem statement, and let:
    \begin{align}
        M &:= \left\lceil  \frac{\log(\delta^{-1}_\text{med})}{2\eta^2} \right\rceil
    \end{align}
         Repeat the above process for a total of $M$ times. If we sum $\vec x$ over $\{0,...,2^{n}-1\}^M$, then we obtain:
            \begin{align}
                \to  \sum_{\vec x} \bigotimes_{l=1}^{M}  \beta\left( \lambda^{(x_l)}_j \right) \ket{k_l} \otimes \ket{\psi_j}
            \end{align}
    \item Run a sorting network, e.g. \cite{1207.2307}, on the $M$ estimates, and copy the median into the output register. 
    \end{enumerate}

The query complexity analysis is straightforward: each estimate requires $2^n - 1$ applications of controlled-$U$, and there are $M$ estimates. So all we must do is demonstrate accuracy. It is clear that the process above implements a map of the form:
    \begin{align}
        \ket{0^n}  \ket{0...0} \ket{\psi_j} \to \sum_{x=0}^{2^n-1}  e^{i\varphi_{j,x}} \sqrt{p_{j,x}} \ket{x} \ket{\text{gar}_{j,k}} \ket{\psi_j} \label{eqn:phaseestmap}
    \end{align}
    for some probabilities $p_{j,x}$ and phases $\varphi_{j,x}$, since this is just a general description of a unitary map that leaves the $\ket{\psi_j}$ register intact. This way of writing the map lets us bound the error in diamond norm to the ideal map
    \begin{align}
        \ket{0^n}  \ket{0...0} \ket{\psi_j} \to e^{i\varphi_{j}}  \ket{\text{floor}(\lambda_j 2^n)} \ket{\text{gar}_{j}} \ket{\psi_j}\label{eqn:phaseestmapideal}
    \end{align}
    (selecting $\varphi_j := \varphi_{j,\text{floor}(\lambda_j 2^n)} $ and $\ket{\text{gar}_{j}}:= \ket{\text{gar}_{j,\text{floor}(\lambda_j 2^n)}} $)
    without uncomputing while still employing the standard analysis of phase estimation which just reasons about the probabilities $p_{j,x}$. These can be bounded from a median amplification analysis of the probabilities $\left| \beta\left(\lambda^{(x)}_j \right)\right|^2$. Consider a vector of $M$ estimates $\vec x$. Then:
    \begin{align}
        p_{j,x} := \sum_{\substack{ \vec x \text{ where} \\ \text{median}(\vec x ) =x }}   \prod_{l=1}^{M}  \left|\beta(  \lambda^{(x_l)}_j )\right|^2 
    \end{align}

    First we show that probabilities $\left| \beta\left(\lambda^{(x)}_j   \right)\right|^2$ associated with the individual estimates are bounded away from $\frac{1}{2}$. Using some identities we can rewrite:
    \begin{align}
        |\beta(\lambda^{(x)}_j)|^2 &= \left|\frac{1}{2^n}\sum_{t=0}^{2^n-1}e^{  2i \pi t \lambda^{(x)}_j /2^n }\right|^2 \\
        &= \left|\frac{1}{2^n} \frac{  e^{  2i \pi \lambda^{(x)}_j  } - 1  }{ e^{  2i \pi \lambda^{(x)}_j / 2^n } - 1   }  \right|^2 \\
        &= \frac{1}{4^n} \frac{ \sin^2( 2\pi \lambda^{(x)}_j ) + (\cos( 2\pi \lambda^{(x)}_j  )-1)^2  }{ \sin^2( 2\pi \lambda^{(x)}_j / 2^n  ) + (\cos( 2\pi \lambda^{(x)}_j / 2^n )-1)^2     }  \\
        &= \frac{1}{4^n} \frac{ 2 -  2\cos( 2\pi \lambda^{(x)}_j  )  }{ 2 -  2\cos( 2\pi \lambda^{(x)}_j / 2^n )  }  \\
        &= \frac{1}{4^n} \frac{ \sin^2( \pi \lambda^{(x)}_j  )  }{ \sin^2( \pi \lambda^{(x)}_j / 2^n )  }  
    \end{align}
    Using the small angle approximation $\sin(\theta) \leq \theta$ for $0 \leq \theta \leq \pi$, we can derive:
    \begin{align}
        |\beta( \lambda^{(x)}_j  )|^2 = \frac{1}{4^n} \frac{ \sin^2( \pi  \lambda^{(x)}_j )  }{ \sin^2( \pi  \lambda^{(x)}_j /2^{n})  }   \geq \frac{1}{4^n} \frac{\sin^2(\pi  \lambda^{(x)}_j )}{ \pi^2 (\lambda^{(x)}_j)^2  / 4^n} = \frac{\sin^2(\pi \lambda^{(x)}_j)}{\pi^2  (\lambda^{(x)}_j)^2 } =: \gamma(\lambda^{(x)}_j )
    \end{align}

\begin{figure}[h]
    \centering
    \includegraphics[width=\textwidth]{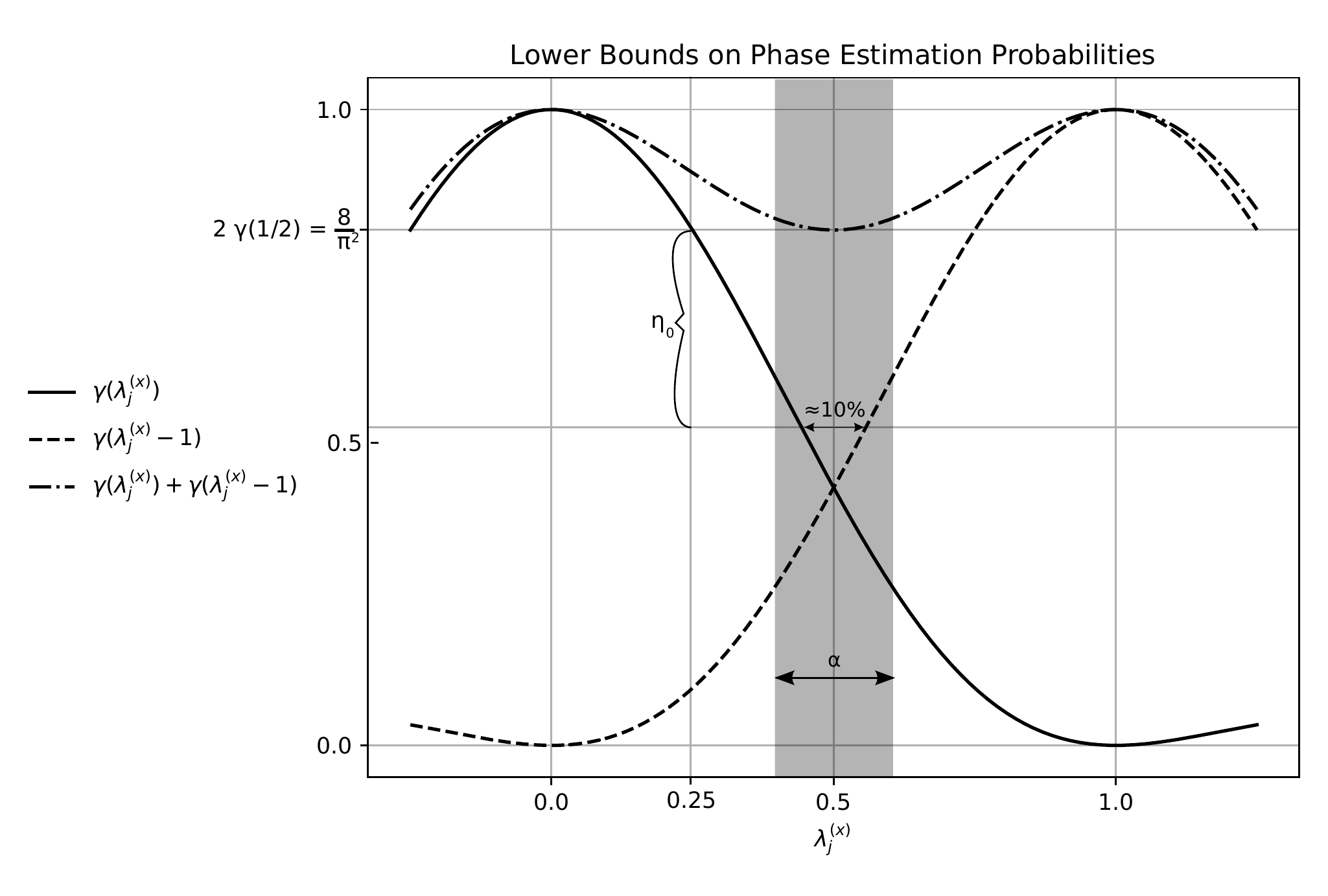}
    \caption{ \label{fig:alpha_eta} Sketch of $\gamma(\lambda^{(x)}_j)$, a lower bound on the magnitude of the amplitude of phase estimation $|\beta(\lambda^{(x)}_j)|^2$. Recall that the rounding promise guarantees that $\lambda_j \not\in \left[\frac{x}{2^n}, \frac{x}{2^n}  + \frac{\alpha}{2^n} \right]$ for all $x$, implying that $\lambda^{(x)}_j \not\in [ \frac{1-\alpha}{2}  ,\frac{1+\alpha}{2}]$ as indicated in the shaded region. From this sketch we draw two conclusions necessary for our proof. First, when $\alpha > 1/2$ then $\gamma(\lambda^{(x)}_j)$ is guaranteed to be strictly greater than 1/2, so $\eta > 0$ (in fact, $\eta > \eta_0$). Second, even when no rounding promise applies, the sum of the probabilities of two adjacent bins is at least $8/\pi^2$, which we use to define the amplification threshold $\eta_0$. Also visible in this figure is the impossibility of reducing $\alpha$ further than $\approx 10\%$ without rounding: several values of $\lambda^{(x)}_j$ near $1/2$ have probabilities less than $1/2$ and cannot be amplified. }
\end{figure}

    $\gamma(\lambda^{(x)}_j )$ is a very tight lower bound to $ |\beta(\lambda^{(x)}_j)|^2$, and is plotted in Figure~\ref{fig:alpha_eta}. We would like this probability to be larger than $\frac{1}{2} + \eta$ for some $\eta$ when $\lambda_j$ satisfies a rounding promise. This is guaranteed if we select:
    \begin{align} 
        \eta := \gamma\left(\frac{1-\alpha}{2}\right) - \frac{1}{2} 
    \end{align}
as in the theorem statement. From the figure, we see that when $\alpha > 1/2$ then $\eta > \eta_0$. This fact actually is not necessary for this construction to work, but observing this will be useful for the $\alpha \leq 1/2$ case. However, it is necessary to observe that when $\alpha > 1/2$ we are guaranteed that $\eta$ is positive.

Then, from the rounding promise and the definition of $\lambda^{(x)}_j$, we are guaranteed that when $x = \text{floor}(2^n \lambda_j)$ then $|\beta(\lambda^{(x)}_j)|^2 \geq 1/2 + \eta$. Now we perform a standard median amplification analysis. The probability of a particular estimate $x$ being `correct' is $\geq 1/2 + \eta $, so the probability of being incorrect is $\leq 1/2 - \eta$. The only way that a median of $M$ estimates can be incorrect is if more than half of the estimates are incorrect. Let $X$ be the random variable counting the number of incorrect estimates. Then we can invoke the Chernoff-Hoeffding theorem:
    \begin{align}
        \text{Pr}[\text{median incorrect}] &\leq \text{Pr}\left[  X \geq M/2 \right]  \\  &\leq \text{Pr}\left[  X \geq M/2 + \mathbb{E}(X) - (1/2 - \eta)M \right] \\ &\leq  \exp\left( - 2M\eta^2 \right)
    \end{align}
    We can bound the above by $\delta_\text{med}$, a constant we will fix later, if we select $M := \left\lceil  \frac{\log(\delta^{-1}_\text{med})}{2\eta^2} \right\rceil$ as we did above.

We have demonstrated that $p_\text{j,x} \geq 1- \delta_\text{med}$ whenever $x = \text{floor}(2^n \lambda_j)$. Now all that remains is a bound on the error in diamond norm. We begin by bounding the spectral norm. Let $p_j := p_{j,\text{floor}(2^n \lambda_j)}$ and observe that $1 \geq p_j \geq 1-\delta_\text{med}$. Taking the difference between equations (\ref{eqn:phaseestmap}) and (\ref{eqn:phaseestmapideal}) we obtain:
\begin{align}
    \label{eqn:phase_calc}  & \left| (\sqrt{p_j} - 1) e^{i\varphi_j}  \ket{\text{floor}(\lambda_j 2^n)} \ket{\text{gar}_{j}} \ket{\psi_j} + \sum_{x,x\neq\text{floor}(\lambda_j 2^n)}e^{i\varphi_{j,x}} \sqrt{p_{j,x}} \ket{x} \ket{\text{gar}_{j,x}} \ket{\psi_j} \right|\\
    = & \sqrt{ |  (\sqrt{p_j} - 1) e^{i\varphi_j}  |^2  + \sum_{x,x\neq\text{floor}(\lambda_j 2^n)} \left| e^{i\varphi_{j,x}}\sqrt{p_{j,x}} \right|^2}\\
    = & \sqrt{  (\sqrt{p_j} - 1)^2  + \sum_{x,x\neq\text{floor}(\lambda_j 2^n)}  p_{j_k} }\\
    \leq & \sqrt{  (  1 - \sqrt{1 - \delta_\text{med}}   )^2  +  \delta_\text{med}  }\\
    \leq & \sqrt{  2  - 2\sqrt{1-\delta_\text{med}} - \delta_\text{med} + \delta_\text{med} }\\
     \label{eqn:phase_calc2}   \leq & \sqrt{  2 - 2\sqrt{1-\delta_\text{med}}  }
\end{align}
Now we apply Lemma~\ref{lemma:spectraltodiamond} to bound the diamond norm, and since the error is $\sim \sqrt{\delta_\text{med}}$ to leading order, we construct a bound that holds whenever $\sqrt{\delta_\text{med}} \leq 1$: 
\begin{align}
    2\sqrt{  2 - 2\sqrt{1-\delta_\text{med}}  } \leq 2.5 \sqrt{\delta_\text{med}}
\end{align}
If we select $\delta_\text{med} := \delta^2/6.25$ then the diamond norm is bounded by $\delta$. 

Having completed the $\alpha > 1/2$ case, we proceed to the $\alpha \leq 1/2$ case. As stated above, the construction we just gave actually also works when $\alpha \leq 1/2$. However, the asymptotic dependence on $\alpha$ is like $\sim\alpha^{-2}$ which is not optimal.

Looking again at Figure~\ref{fig:alpha_eta}, we see that the sum of the probability of two adjacent bins is at least $8/\pi^2$, even when $\lambda^{(x)}_j$ is in a region disallowed by the rounding promise. Therefore, if we perform median amplification with gap parameter $\eta_0$ we are guaranteed that values of $\lambda^{(x)}_j$ that fall into a rounding gap will at least be rounded to an adjacent bin.

We can use this fact to achieve an $\sim \alpha^{-1}$ dependence. Recall that $\alpha$ is the fraction of the range $[0,1)$ where $\lambda_j$ are not allowed to appear due to the rounding promise, independent of $n$. If we perform phase estimation with $n+1$ rather than $n$, then the width of each disallowed region is cut in half from $\alpha 2^{-n}$ to $\alpha 2^{-n-1}$ - but we also double the number of disallowed regions. However, if we simply ignore the final bit, then, because we are amplifying to $\eta_0$, any values of $\lambda_j$ that fall into one of the newly introduced regions will simply be rounded. Thus, we cut the gaps in half without introducing any new gaps, so $\alpha$ is cut in half. We will make this argument more formal in a moment.
 
 So, in order to achieve a particular $\alpha$, we observe from Figure~\ref{fig:alpha_eta} that if we amplify to $\eta_0$ then the gaps have width $\frac{1}{2} \cdot 2^{-n}$. If we estimate an additional $r$ bits, then the gaps have width $\frac{1}{2} \cdot 2^{-n-r}$. Solving $\frac{1}{2} \cdot 2^{-n-r} \leq \alpha 2^{-n}$ for $r$, we obtain
 \begin{align}
r := \left\lceil \log_2 \left(\frac{1}{2\alpha}\right)  \right\rceil.
    \end{align}

We briefly make the $n$ explicit in $\lambda^{(x)}_j$ by writing $\lambda^{(n,x)}_j$. After running the protocol at the beginning of the proof with $n+r$ bits, we obtain, for $\vec x \in \{0,...,2^{n}-1\}^M$, the state:
    \begin{align}
        &\sum_{\vec x} \ket{\text{median}(\vec x)}  \otimes  \bigotimes_{l=1}^{M} \sum_{t=0}^{2^r - 1} \beta(  \lambda_j^{(n+r,\hspace{1mm} 2^{r} x_l + t )} ) \ket{2^{r}k_l + t} \otimes \ket{\psi_j}\\
        = & \sum_{x=0}^{2^n-1} \ket{x} \otimes \sum_{\substack{ \vec x \text{ where} \\ \text{median}(\vec x ) =kx}}   \bigotimes_{l=1}^{M} \sum_{t=0}^{2^r - 1} \beta(  \lambda_j^{(n+r,\hspace{1mm}  2^{r} x_l + t )} ) \ket{2^{r}x_l + t} \otimes \ket{\psi_j}\\
        = & \sum_{x=0}^{2^n-1} \sqrt{p_{j,x}} e^{i\varphi_{j,x}}\ket{x} \otimes \ket{\text{gar}_{j,x}} \otimes \ket{\psi_j}
    \end{align}
    where in the last step we have defined the normalized state $\ket{\text{gar}_x}$, the probability $p_{j,x}$, and the phase $\varphi_{j,x}$ via:
    \begin{align}
        \sqrt{p_{j,x}} e^{i\varphi_{j,x}} \ket{\text{gar}_{j,x}} :=  \sum_{\substack{ \vec x \text{ where} \\ \text{median}(\vec x ) =x }}   \bigotimes_{l=1}^{M} \sum_{t=0}^{2^r - 1} \beta(  \lambda_j^{(n+r,\hspace{1mm}  2^{r} x_l + t )} ) \ket{2^{r}x_l + t} 
    \end{align}
Now if we can demonstrate that $p_{j,x} \geq 1- \delta_\text{med}$ when $x = \text{floor}(2^n\lambda_j)$ then the same argument as in (\ref{eqn:phase_calc}-\ref{eqn:phase_calc2}) holds and we are done. Observe that:
 \begin{align}
     p_{j,x} := \sum_{\substack{ \vec x \text{ where} \\ \text{median}(\vec x ) =x }}   \prod_{l=1}^{M} \sum_{t=0}^{2^r-1}  \left|\beta(  \lambda_j^{(n+r,\hspace{1mm}  2^r x_l + t  )} )\right|^2 
    \end{align}

Say $x = \text{floor}(2^n \lambda_j)$. Then, consider $t := \text{floor}(2^{n+r} \lambda_j) - x 2^{r} $, which is an integer $\in \{0,...,2^{r}-1\}$. Then, looking at  Figure~\ref{fig:alpha_eta}, we see that:
\begin{align}
    |\beta(\lambda_j^{(n+r,\hspace{1mm} 2^r x + t)})|^2  + |\beta(\lambda_j^{(n+r,\hspace{1mm} 2^r x + t + 1)})|^2  \geq \gamma(\lambda_j^{(n+r,\hspace{1mm} 2^r x + t)}) + \gamma(\lambda_j^{(n+r,\hspace{1mm} 2^r x + t)} - 1) \geq \frac{1}{2} + \eta_0
\end{align}

Therefore the probability that the first $n$ bits are $x$ is $\sum_{t=0}^{2^r-1}  \left|\beta(  \lambda_j^{(n+r,\hspace{1mm}  2^r x + t  )} )\right|^2 \geq 1/2 + \eta_0$, which is all that was required to perform the median amplification argument above. So we can conclude that $p_{j,k} \geq 1- \delta_\text{med}$, so the modified algorithm retains the same accuracy.
\end{proof}

\section{Coherent Iterative Estimation}  \label{sec:coherent}

In this section we present the novel algorithms for phase estimation and energy estimation. As described in the introduction, both of these feature a strong similarity to `iterative phase estimation' \cite{9511026}, where the bits of the estimate are obtained one at a time. Unlike iterative phase estimation however, the state is never measured and the entire process is coherent. We therefore name these algorithms as `coherent iterative estimators'.

Another similarity that these new algorithms share with the original iterative phase estimation is that the less significant bits are taken into account when obtaining the current bit. This greatly reduces the amount of amplification required for the later bits, so the runtime is vastly dominated by the estimation of the least significant bit. We will go into more detail on this later.

To permit discussion of coherent iterative estimation of phases and energies in a unified manner, we fit this idea into the modular framework of Definition~\ref{def:estimator} and Lemma~\ref{lemma:garbageremove}. A `coherent iterative estimator' obtains a single bit of the estimate, given access to all the previous bits. Several invocations of a coherent iterative estimator yield a regular estimator as in Definition~\ref{def:estimator}. Furthermore, we can choose to uncompute the garbage at the very end using Lemma~\ref{lemma:garbageremove}, or, as we will show, we can remove the garbage early, which prevents it from piling up.

\begin{definition} \label{def:iterativeangle} A \textbf{coherent iterative phase estimator} is a protocol that, given some $n$, $\alpha$, any $k \in \{0,...,n-1\}$, and any error target $\delta > 0$, produces a quantum circuit involving calls to controlled-$U$ and $U^\dagger$. If the unitary $U$ satisfies an $(n,\alpha)$-rounding promise, then this circuit implements a quantum channel that is $\delta$-close to some map that performs:
    \begin{align}
        \ket{0} \ket{\Delta_k} \ket{\psi_j}  \to \ket{\text{bit}_k\left( \lambda_j \right)}\ket{\Delta_k}\ket{\psi_j}
    \end{align}
    Here $\text{bit}_k(\lambda_j)$ is the $(k+1)$'th least significant bit of an $n$-bit binary expansion, and $\ket{\Delta_k}$ is a $k$-qubit register encoding the $k$ least significant bits. (For example, if $n=4$ and $\lambda_j = 0.1011...$ then $\text{bit}_2(\lambda_j) = 0$ and $\Delta_2 = 11$.) Note that the target map is only constrained on a subspace of the input Hilbert space, and can be anything else on the rest.

    An \textbf{coherent iterative energy estimator} is the same thing, just with controlled-$U_H$ and $U_H^\dagger$ queries to some block-encoding $U_H$ of a Hamiltonian $H$ that satisfies an $(n,\alpha)$-rounding promise. 

    A coherent iterative estimator can also be with garbage and/or with phases, just like in Definition~\ref{def:estimator}.
\end{definition}


\begin{lemma} \label{lemma:stitchingiterative} \textbf{Stitching together coherent iterative estimators.} Given a coherent iterative phase/energy estimator with query complexity $Q(n,k',\alpha,\delta')$, we can construct a non-iterative phase/energy estimator with query complexity:
    \begin{align}
        \sum_{k=0}^{n-1} Q\left(n,k,\alpha, \delta\cdot 2^{-k-1} \right) 
    \end{align}
The non-iterative estimator has phases if and only if the coherent iterative estimator has phases, and if the coherent iterative estimator has $m$ qubits of garbage then the iterative estimator has $nm$ qubits of garbage.
\end{lemma}
\begin{proof} We will combine $n$ many coherent iterative phase/energy estimators, for $k = 0,1,2,..,n-1$. The diamond norm satisfies a triangle inequality so if we let the $k$'th iterative estimator have an error $\delta_k := \delta 2^{-k-1}$ then the overall error will be:
    \begin{align}
        \sum_{k=0}^{n-1} \delta_k \leq \frac{\delta}{2}\sum_{k=0}^{n-1} 2^{-k} = \frac{\delta}{2} (2 - 2^{1-n}) \leq \delta
    \end{align}
So now all that is left is to observe that the exact iterative estimators chain together correctly. This should be clear by observing that for all $k > 0$:
    \begin{align}
        \ket{\Delta_k} =  \ket{ \text{bit}_{k-1} (\lambda_j) }\otimes ... \otimes \ket{ \text{bit}_{0} (\lambda_j)}
    \end{align}
and $\ket{\Delta_0} = 1 \in \mathbb{C}$ since when $k=0$ there are no less significant bits. So the $k$'th iterative estimator takes the $k$ least significant bits as input and computes one more bit, until finally at $k = n-1$ we have:
    \begin{align}
        \ket{ \text{bit}_{n-1} (\lambda_j) } \ket{\Delta_{n-1}} =  \ket{ \text{bit}_{n-1} (\lambda_j) }\otimes ... \otimes \ket{ \text{bit}_{0} (\lambda_j)} = \ket{\text{floor}(2^n \lambda_j)}
    \end{align}
The total query complexity is just the sum of the $n$ invocations of the iterative estimators from $k=0,...,n-1$ with error $\delta_k$.

    If the iterative estimator has garbage, then the garbage from each of the $n$ invocations just piles up. Similarly, if the estimator is with phases, and has an $e^{i\varphi_{j,k}}$ for the $k$'th invocation, then the composition of the maps will have a phase $\prod_{k=0}^{n-1}e^{i\varphi_{j,k}}$.
\end{proof}


\begin{lemma} \label{lemma:iterativegarbageremove} \textbf{Removing garbage and phases from iterative estimators.} Given a coherent iterative phase/energy estimator with phases and/or garbage and with query complexity $Q(n,k,\alpha,\delta')$ that has a unitary implementation, we can construct a coherent iterative phase/energy estimator without phases and without garbage with query complexity $2Q(n,k,\alpha,\delta/2)$. \end{lemma}
    \begin{proof} This argument proceeds exactly the same as Lemma~\ref{lemma:garbageremove}, just with the extra $\ket{\Delta_k}$ register trailing along. If $\Lambda$ is the channel that implements the iterative estimator with garbage and/or phases, then we obtain an estimator without garbage and phases via:
\begin{align}
       \hspace{5mm} \begin{array}{c}\Qcircuit @C=1em @R=1em {
                    \lstick{\ket{0}}   & \qw & \targ & \qw  & \qw\\
                \lstick{\ket{0}} & \multigate{3}{\Lambda} & \ctrl{-1} & \multigate{3}{\Lambda^{-1}} &  \measuretab{\text{discard}}\\
           \lstick{\ket{0...0}}     & \ghost{\Lambda} & \qw & \ghost{\Lambda^{-1}} & \measuretab{\text{discard}}\\
           \lstick{\ket{\Delta_k}}    & \ghost{\Lambda} & \qw & \ghost{\Lambda^{-1}} & \qw \\
           \lstick{\ket{\phi_j}}    & \ghost{\Lambda} & \qw & \ghost{\Lambda^{-1}} & \qw 
       }\end{array} 
\end{align}
\end{proof}


The advantage of the modular framework we just presented is that maximizes the amount of flexibility when implementing these algorithms. How exactly uncomputation is performed will vary from application to application, and depending on the situation uncomputation may be performed before invoking Lemma~\ref{lemma:stitchingiterative} using Lemma~\ref{lemma:iterativegarbageremove}, after Lemma~\ref{lemma:stitchingiterative} using Lemma~\ref{lemma:garbageremove}, or not at all!

Of course, the idea of uncomputation combined with iterative estimation itself is quite simple, so given a complete understanding of the techniques we present the reader may be able to perform this modularization themselves. However, we found that this presentation significantly de-clutters the presentation of the main algorithms, those that actually implement the coherent iterative estimators from Definition~\ref{def:iterativeangle}. While the intuitive concept behind these strategies is not so complicated, the rigorous presentation and error analysis is quite intricate. We therefore prefer to discuss uncomputation separately.

\subsection{Coherent Iterative Phase Estimation}   \label{sec:phaseestimation}

This section describes our first novel algorithm, presented in Theorem~\ref{thm:iterativephaseestimation}. Before stating the algorithm in complete detail, performing an error analysis, and showing how to optimize the performance up to constant factors, we outline the tools that we will need and give an intuitive description.

As stated in the introduction, a very similar algorithm was independently discovered by \cite{2105.02859}. The techniques of the two algorithms for phase estimation feature some minor differences: our work is more interested in maintaining coherence of the input state, the algorithm's constant-factor runtime improvement over prior art, and our runtime is $O(2^n)$ rather than $O(n2^n)$ ($O(n)$ vs $O(n \log n)$ respectively in the language of \cite{2105.02859}). On the other hand, \cite{2105.02859} elegantly show how a quantum Fourier transform emerges as a special case of the algorithm, and their presentation is significantly more accessible. Both methods avoid use of ancillae entirely.

A key tool for these algorithms is the block-encoding, which allows us to manipulate arbitrary non-unitary matrices. In this paper we simplify the notion a bit by restricting to square matrices with spectral norm $\leq 1$. 


\begin{definition} \label{def:blockencoding} Say $A$ is a square matrix $\in\mathcal{L}\left(\mathcal{H}\right)$. A unitary matrix $U_A$ is a \textbf{block-encoding} of $A$ if $U_A$ acts on $m$ qubits and $\mathcal{H}$, and:
    \begin{align}
        \left(\bra{0^m} \otimes I\right) U \left( \ket{0^m} \otimes I \right) = A
    \end{align}
    We can also make $m$ explicit by saying that $A$ \textbf{has a block-encoding with $m$ ancillae}. We allow the quantum circuit implementing $U$ to allocate ancillae in the $\ket{0}$ state and then to return them to the $\ket{0}$ state with probability 1. These ancillae are not postselected and do not contribute to the ancilla count $m$.  
\end{definition}


A unitary matrix is a trivial block-encoding of itself. In this sense, we already have a block-encoding of the matrix:
\begin{align}
    U = \sum_j e^{2\pi i \lambda_j} \ket{\psi_j}\bra{\psi_j}
\end{align}

Recall that the goal of the coherent iterative estimator is to compute $\text{bit}_k(\lambda_j)$. The strategy involves preparing an approximate block-encoding of:
\begin{align}
    \sum_j \text{bit}_k(\lambda_j) \ket{\psi_j}\bra{\psi_j}
\end{align}

We begin by rewriting the above expression a bit. Recall that $\Delta_k$ is an integer from $0$ to $2^{k+1}-1$ encoding the $k$ less significant bits of $\lambda_j$. If we subtract $\Delta_k/2^n$ from $\lambda_j$, we obtain a multiple of $1/2^{n-k}$ plus something $< 1/2^{n}$ which we floor away. Then, $\text{bit}_k(\lambda_j)$ indicates if this is an even or an odd multiple. That means we can write:
\begin{align}
    \text{bit}_k(\lambda_j) = \text{parity}\left(\text{floor}\left( 2^{n-k} \left( \lambda_j - \frac{\Delta_k}{2^n} \right) \right)\right)
\end{align}

Since $ \lambda_j - \frac{\Delta_k}{2^n}$ is a multiple of $1/2^{n-k}$ plus something $< 1/2^{n}$, we equivalently have that $2^{n-k} \left( \lambda_j - \frac{\Delta_k}{2^n} \right)$ is an integer plus something less than $1/2^{k}$. For such values, we can write the parity(floor($x$)) function in terms of a squared cosine that has been `amplified':
\begin{align}
    \text{parity}(\text{floor}(x)) = \text{amp}\left(\cos^2\left( \frac{\pi}{2} \left[ x + \phi  \right] \right)\right)\\
    \text{amp}(x) = \Bigg\{ \begin{array}{cc}1 \text{ if } x > 1/2\\ 0 \text{ if } x < 1/2 \end{array}
\end{align}
where the shift $\phi$ centers the extrema of the cosine in the intervals where $x$ occurs. Therefore:
\begin{align}
    \text{bit}_k(\lambda_j) &= \text{amp}\left(\cos^2\left( \frac{\pi}{2} \left[ 2^{n-k} \left(\lambda_j - \frac{\Delta_k}{2^n}\right) + \phi \right]  \right)\right)  \\
    &= \text{amp}\left(\cos^2\left( \pi \left[ 2^{n-k-1} \left(\lambda_j - \frac{\Delta_k}{2^n}\right) + \phi/2 \right]  \right)\right)  \\
    &= \text{amp}\left(\cos^2\left( \pi \lambda^{(k)}_j \right)\right) 
\end{align}
Where we have defined:
\begin{align}
    \lambda_j^{(k)} := 2^{n-k-1} \left( \lambda_j - \frac{\Delta_k}{2^n} \right) + \phi_k
\end{align}
for some $k$-dependent choice of phase $\phi_k = \phi/2$.

By applying a phase shift conditioned on the $\ket{\Delta_k}$ register, and then iterating it $2^{n-k-1}$ times, we can construct a block-encoding of the unitary with eigenvalues $\lambda_j^{(k)}$. Our goal is now to transform this unitary as follows:
\begin{align}
    \sum_j e^{2\pi i \lambda_j^{(k)}} \ket{\psi_j}\bra{\psi_j} \hspace{5mm}\to \hspace{5mm}  \sum_j \left[ \text{amp}\left(\cos^2\left( \pi \lambda^{(k)}_j \right)\right) \right]\ket{\psi_j}\bra{\psi_j}
\end{align}
To obtain a cosine use linear combinations of unitaries \cite{1501.01715, 1511.02306} to take a take a linear combination with the identity:
\begin{align}
    \sum_j \frac{e^{2 \pi i \lambda^{(k)}_j} + 1}{2} \ket{\psi_j}\bra{\psi_j}  &= \sum_j \cos\left( \pi \lambda_j^{(k)} \right) \left[ e^{i \pi \lambda_j^{(k)} } \ket{\psi_j}\right]\bra{\psi_j}\\
    &= \sum_j \left|\cos\left(\pi \lambda_j^{(k)} \right)\right| \cdot \left[ \pm e^{i \pi \lambda_j^{(k)} } \ket{\psi_j}\right]\bra{\psi_j}
\end{align}
where on the previous line $\pm$ indicates $\text{sign}\left( \cos\left( \pi \lambda^{(k)}_j \right) \right)$. This way the above is a singular value decomposition of the block-encoded matrix.
So all that is left to do is to approximately transform the singular values $a$ of the matrix above by:
\begin{align}
    a \to \text{amp}(a^2)
\end{align}
We will accomplish this using singular value transformation.


\begin{lemma} \textbf{Singular value transformation.}\label{lemma:svt}  Say $A$ is a square matrix with singular value decomposition $A = \sum_i a_i \ket{\psi^l_i}\bra{\psi^r_i}$, and $p(x)$ is a degree-$d$ even polynomial with real coefficients satisfying $|p(x)| \leq 1$ for $|x| \leq 1$, and suppose $A$ has a block-encoding $U_A$. Then, for any $\delta > 0$, there exists a $O(\text{poly}(d,\log(1/\delta)))$ time classical algorithm that produces a circuit description of a block-encoding of the matrix:
        \begin{align}
            \tilde p(A) := \sum_i \tilde p(a_i) \ket{\psi^r_i}\bra{\psi^r_i}
        \end{align}
        where $\tilde p(x)$ is a polynomial satisfying $|\tilde p(x) - p(x)| \leq \delta$ for $|x| \leq 1$.  This block-encoding makes $d$ queries to $U_A$ or $U_A^\dagger$ (not controlled). If the block-encoding of $A$ has $m$ ancillae, then that of $\tilde p(A)$ has $m+1$ ancillae. In the special case when $m=1$ and the unitary implementing the block-encoding has the form $U_A = \sum_i V_i \otimes  \ket{\psi^l_i}\bra{\psi^r_i}$ where the $V_i$ are qubit reflections, then it is possible to make the block-encoding of $\tilde p(A)$ also just have one ancilla.
\end{lemma}
\begin{proof} This is a slightly simplified version of the main result of \cite{1806.01838}. The $m>1$ case involves an application of Corollary~18 which demands an extra control qubit that is also postselected. For an accessible review on the subject see \cite{2105.02859}.

    We briefly elaborate on the $m=1$ special case when $U_A = \sum_i V_i \otimes  \ket{\psi^l_i}\bra{\psi^r_i}$ and the $V_i$ are reflections. We find that when this is the case, the $(W_x,S_z, \bra{+} \cdot \ket{+})$-QSP convention (see Theorem 13 of \cite{2105.02859}) can actually be implemented by converting to the reflection convention and then using the circuit from Lemma~19 of \cite{1806.01838}, simply by applying a Hadamard gate to the ancilla register at the beginning and end of the circuit. This lets us implement polynomials with the constraints above without resorting to Corollary~18 of \cite{1806.01838}.
    
    Note that the circuit from Lemma~19 of \cite{1806.01838} (which is used in both the $m=1$ and $m>1$ cases) initializes an extra qubit to perform reflections. However, this qubit is guaranteed to be returned to the $\ket{0}$ state exactly and is not postselected, so it is not an ancilla in the sense of Definition~\ref{def:blockencoding}.
\end{proof}


Singular value transformation can perform the desired conversion if we can construct a polynomial $A(x)$ such that:
\begin{align}
    A(x) \approx \text{amp}(x) = \frac{1}{2} - \frac{1}{2} \text{sign}(2x - 1)
\end{align}
If we invoke the above lemma with the even polynomial $A(x^2)$ we get an approximation to the desired block-encoding of $\sum_j \text{bit}_k(\lambda_j)\ket{\psi_j}\bra{\psi_j}$.

In particular, the behavior we must capture in the polynomial approximation is that $A(x) \approx 1$ when $x \in [0,1/2 - \eta]$ and $A(x) \approx 0$ when $x \in [1/2 + \eta,1]$ for some gap $\eta$ away from $1/2$. If we view the input $x$ as a probability, then $A(x)$ essentially `amplifies' this probability to something close to 0 or 1 (and additionally flips the outcome). We hence call $A(x)$ an `amplifying polynomial'.

One approach to constructing such an amplifying polynomial is to simply adapt a classical algorithm for amplification. We stated this method in the introduction: the polynomial is the probability that the majority vote of several coin tosses is heads, where each coin comes up heads with probability $x$. Then the desired properties can be obtained from the Chernoff-Hoeffding theorem \cite{0902.3757}. The number of coins we have to toss to accomplish a particular $\eta,\delta$ is the degree of the polynomial, which is bounded by $O(\eta^{-2} \log(\delta^{-1}))$. 

However, this polynomial does not achieve the optimal $\eta$ dependence of $O(\eta^{-1})$. This might be achieved by a polynomial inspired by a quantum algorithm for approximate counting, which does achieve the $O(\eta^{-1})$ dependence. But rather than go through such a complicated construction we simply adapt a polynomial approximation to the $\text{sign}(x)$ function developed in \cite{1707.05391}, which accomplishes the optimal $O(\eta^{-1} \log(\delta^{-1}))$. 


\begin{lemma} \label{lemma:amppoly} \textbf{Quantum amplifying polynomial.} For any $0  < \eta,\delta < 1/2 $, let:
        \begin{align}
            I_j(t) &:= \text{the } j\text{'th modified Bessel function of the first kind}\\
            T_j(t) &:= \text{the } j\text{'th Chebyshev polynomial of the first kind}\\
            k &:= \frac{\sqrt{2}}{4\eta} \sqrt{  \ln\left( \frac{8}{\pi \delta}  \right)  } 
        \end{align}
   For some $M_{\eta\to\delta}$, consider the polynomials:
    \begin{align}
        p_\text{sgn}(x) &:= \frac{2ke^{-\frac{k^2}{2}}}{\sqrt{\pi}} \left( I_0\left(\frac{k^2}{2}\right) \cdot x + \sum_{j=1}^{\frac{1}{2}M_{(\eta\to\delta)}-\frac{1}{2}} I_j\left(\frac{k^2}{2}\right)  (-1)^j \left( \frac{T_{2j+1}\left(x\right)}{2j+1} - \frac{ T_{2j-1}\left(x\right)  }{2j-1}   \right) \right)\\
        A_{\eta\to\delta}(x)&:= \frac{1}{2} - \frac{1}{2} \frac{p_{\text{sgn}}(2x-1)}{1+\delta/2}   
    \end{align}
    Then there exists an $M_{\eta\to\delta} \in O\left( \eta^{-1} \log\left( \delta^{-1}\right) \right)$ such that $A_{\eta\to\delta}(x)$ is of degree $M_{\eta\to\delta}$ and satisfies the constraints:
    \begin{align}
        \label{eqn:amppoly1}  \text{ if } 0 \leq x \leq 1 &\text{ then } 0 \leq A_{\eta\to\delta}(x) \leq 1\\
        \label{eqn:amppoly2} \text{ if } 0 \leq x \leq \frac{1}{2} - \eta  &\text{ then }     A_{\eta\to\delta}(x)  \geq 1- \delta\\
        \text{ if } \frac{1}{2} + \eta \leq x \leq 1 &\text{ then }     A_{\eta\to\delta}(x)  \leq  \delta
    \end{align}
\end{lemma}
\begin{proof} This is a more self-contained re-statement of Corollary~6 in Appendix~A of \cite{1707.05391}, which constructs a polynomial $p_{\text{sgn},\kappa,\delta,n}(x) \approx \text{sign}(x)$ with various accuracy parameters. The polynomial above is $\frac{1}{2} - \frac{1}{2} p_{\text{sgn},\kappa,\delta/2,n}\left(2x-1\right)$, since after all we desired $A(x) \approx \frac{1}{2} - \frac{1}{2}\text{sign}(2x-1)$. To guarantee good behavior when $|x-1/2| > \eta$, we select $\kappa := 4\eta$. The value of $M_{\eta\to\delta}$ itself can be computed by combining various results from that work's Appendix~A.
\end{proof}


The method for obtaining $M_{\eta\to\delta}$ given $\eta$ and $\delta$ is complicated enough that it is not worth re-stating here. However, the complexity is by all means worth it: in our numerical analyses we found that the polynomials presented in Appendix~A of \cite{1707.05391} feature excellent performance in terms of degree. This is a major source of the query complexity speedup of our algorithms.

The value of $\delta$ is chosen such that the final error in diamond norm is bounded. The value of $\eta$ depends on how far away $\cos^2\left( \pi \lambda^{(k)}_j\right)$ is from $\frac{1}{2}$. Of course, there are several possible values of $\lambda^{(k)}_j$ where $\cos^2\left( \pi \lambda^{(k)}_j\right) = \frac{1}{2}$ exactly, so $\eta = 0$ and amplification is impossible.  This is where the rounding promise comes in: it ensures that $\lambda_\Delta$ is always sufficiently far from such values, so that we can guarantee that $\cos^2\left( \pi \lambda^{(k)}_j\right)$ is always either $\geq \frac{1}{2} + \frac{\alpha}{2}$ or $\leq \frac{1}{2} + \frac{\alpha}{2}$. So we select $\eta = \alpha/2$.

When $k = 0$ then indeed the rounding promise is the only thing guaranteeing that amplification will succeed. However, if the less significant bits have already been computed then the set of values that $\lambda_j^{(k)}$ can take is restricted. This is because the previous bits of $\lambda_j$ have been subtracted off. This widens the region around the solutions of $\cos^2\left(\pi \lambda^{(k)}_j\right) = \frac{1}{2}$ where $\lambda_\Delta$ cannot be found, allowing us to increase $\eta$. Furthermore, this can be done without relying on the rounding promise anymore: bits with $k \geq 1$ are guaranteed to be deterministic even if no rounding promise is present. That means that if the rounding promise is violated, then only the least significant bit can be wrong. The polynomials are sketched in Figure~\ref{fig:projectionpolys}.

\begin{figure}[t]
    \centering
    \includegraphics[width=\textwidth]{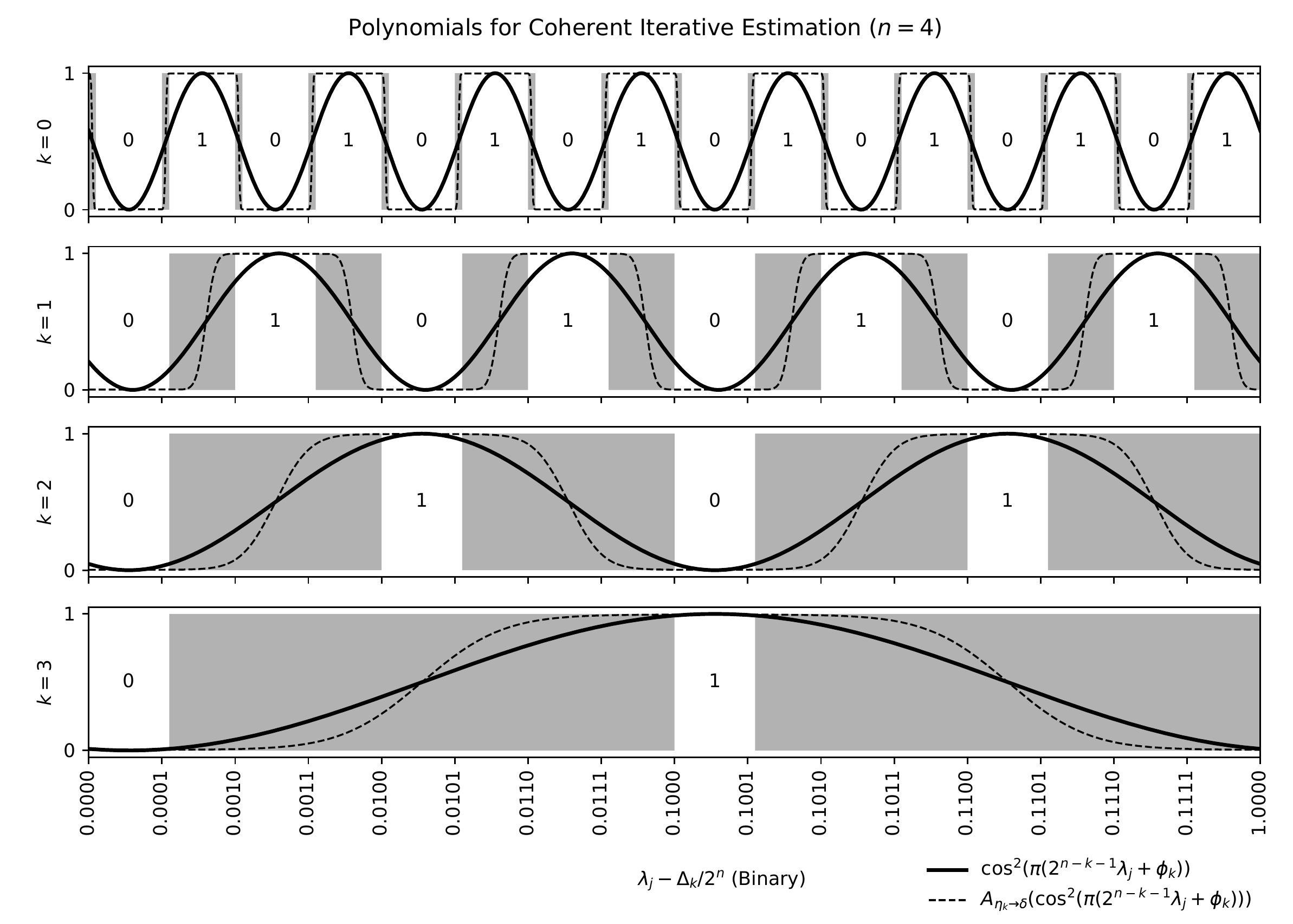}
    \caption{  \label{fig:projectionpolys}Sketch of the polynomials used in Theorems~\ref{thm:iterativephaseestimation}~and~\ref{thm:iterativeenergyestimation}. In black we show a shifted $\cos^2(x)$ function that indicates if the bit is 0 or 1. Then, the amplifying polynomial from Lemma~\ref{lemma:amppoly} is applied to it to yield the dashed line, which is either $\leq \delta$ or $\geq 1-\delta$ depending on the bit. For the $k=0$ bit, the gaps between the allowed intervals are guaranteed by the rounding promise. But as more bits are estimated and subtracted off, the gaps for $k \geq 1$ require no rounding promise, and also become larger and larger so less and less amplification is needed. }
\end{figure}


After constructing the amplifying polynomial $A_{\eta\to\delta}$, we use singular value transformation to apply $A_{\eta\to\delta}(x^2)$ which is even as required by Lemma~\ref{lemma:svt}. Now we have an approximate block-encoding of $\sum_j \text{bit}_k(\lambda_j) \ket{\psi_j}\bra{\psi_j}$, which is in fact a projector. In the introduction we stated that we would use a block-measurement theorem to compute the map:
\begin{align}
        \ket{0}\otimes \ket{\psi} \to   \ket{1} \otimes \Pi\ket{\psi} + \ket{0} \otimes (I-\Pi)\ket{\psi}
\end{align}
given an approximate block-encoding of $\Pi$. However, this general tool involves uncomputation which we specifically wanted to modularize. The fact that we are already measuring errors in terms of the diamond norm means that Lemmas~\ref{lemma:garbageremove}~and~\ref{lemma:iterativegarbageremove} are already capable of dealing with garbage. We therefore defer the proof of this general tool to Section~\ref{sec:blockmeas}, specifically Theorem~\ref{thm:blockmeasure}.

There is another reason to not use Theorem~\ref{thm:blockmeasure} in a black-box fashion, specifically for coherent iterative phase estimation. The block-encoding of $\sum_j \text{bit}_k(\lambda_j) \ket{\psi_j}\bra{\psi_j}$ actually features only one ancilla qubit: the qubit we used to take the linear combination of $U$ and the identity. That means that the block-encoding itself is already very close to the map $\ket{0}\otimes \ket{\psi} \to   \ket{0} \otimes \Pi\ket{\psi} + \ket{1} \otimes (I-\Pi)\ket{\psi}$ (note the flipped output qubit). The details of this usage of the block-encoding will appear in the proof.

This completes the sketch of the procedure to implement the map
\begin{align}
    \ket{0} \ket{\Delta_k} \ket{\psi_j}  \to e^{i\varphi_j} \ket{\text{bit}_k\left( \lambda_j \right)}\ket{\Delta_k}\ket{\psi_j}
\end{align}
for some phases $\varphi_j$. We can now state the protocol in detail, and perform the accuracy analysis.

\begin{theorem} \label{thm:iterativephaseestimation} \textbf{Coherent Iterative Phase Estimation.} There is a coherent iterative phase estimator with phases (and no garbage) and with query complexity:
    \begin{align} 2^{n-k} \cdot M_{\eta\to\delta_\text{amp}}   \end{align}
        where in the above $M_{\eta\to\delta}$ is as in Lemma~\ref{lemma:amppoly},  $\eta := \frac{1}{2} - \frac{1}{2^k} \left( \frac{1}{2} + \frac{\alpha}{2} \right) $ if  $ k \geq 1$ and $\eta := \frac{\alpha}{2}$ if $k=0$, and $\delta_\text{amp}$ can be chosen to be $ (1-10^{-m}) \delta^2 / 8 $ for any $m > 0$. 
\end{theorem}

\begin{proof} We construct the estimator as follows:
    \begin{enumerate}
        \item Construct a unitary that performs a phase shift depending on $\Delta_k$ - we call this $e^{-2\pi i \hat\Delta_k/2^n}$ employing some notation inspired by physics literature.
        \begin{align}
            e^{-2 \pi i \hat\Delta_k/2^n} := \sum_{\Delta_k=0}^{2^k-1} e^{-2\pi i \Delta/2^n}\ket{\Delta_k}\bra{\Delta_k} = \bigotimes_{j=0}^{k-1}  e^{ -2\pi i \pi 2^{j-n}} 
        \end{align}
    \item Rewrite the oracle unitary in this notation:
        \begin{align}
            e^{2\pi i \hat\lambda} := U = \sum_{j} e^{2\pi i\lambda_j}  \ket{\psi_j}\bra{\psi_j}  \label{eqn:lambdakunitary}
        \end{align}
            Then define: 
            \begin{align}
                \phi_0 &:= 1 - \text{mean}\left( \frac{1}{2} + \frac{\alpha}{2} ,  1 \right)\\
                \phi_k &:= 1 - \text{mean}\left( \frac{1}{2} , \frac{1}{2} + \frac{1}{2^k} \left( \frac{1}{2} + \frac{\alpha}{2} \right) \right)  \text{ if } k \geq 1\\
                \hat\lambda^{(k)} &:= 2^{n-k-1} \left(\hat\lambda - \frac{\hat\Delta_k}{2^n}\right) + \phi_k,
            \end{align}
            and implement the corresponding phase shift:
        \begin{align}
            e^{2\pi i \hat\lambda^{(k)}} := \left(e^{-2\pi i\hat\Delta_k/2^n} \otimes  e^{2\pi i\hat\lambda}\right)^{2^{n-k-1}} \cdot e^{2\pi i \phi_k}
        \end{align} 
            This unitary acts jointly on the $\ket{\Delta_k}$ and $\ket{\psi_j}$ inputs. Since $k \in \{0,...,n-1\}$, the exponent $2^{n-k-1}$ is always an integer.
        \item Let $\tilde H := \frac{1}{\sqrt{2}} \begin{bmatrix} 1 & 1 \\ i & -i \end{bmatrix} $ be a slightly modified Hadamard gate, and consider the following unitary, implemented via $e^{2\pi i \hat\lambda^{(k)}}$:
\begin{align}
    U^{(k)}_\text{signal}  := \hspace{10mm} \begin{array}{c}\Qcircuit @C=1em @R=1em {
                      & \gate{\tilde H} & \ctrl{1} & \gate{\tilde H^T} & \qw \\
           \lstick{\ket{\Delta_k}} & \qw & \multigate{1}{e^{2 \pi i \hat\lambda^{(k)}}} & \qw & \qw\\
           \lstick{\ket{\psi_j}}   & \qw & \ghost{e^{ 2  \pi i \hat\lambda^{(k)}}}         & \qw  & \qw
    }\end{array} \label{eqn:usignal_def}
\end{align}
         Observe that:
            \begin{align}
                U^{(k)}_\text{signal} &= \sum_{\Delta_k} \sum_j \tilde H \begin{bmatrix} \hspace{4mm}1\hspace{4mm} & 0 \\ 0 & e^{ 2 \pi i \lambda_j^{(k)}}   \end{bmatrix} \tilde H^T \otimes \ket{\Delta_k}\bra{\Delta_k} \otimes \ket{\psi_j}\bra{\psi_j}\\
                    &= \sum_{\Delta_k} \sum_j  e^{ \pi i \lambda^{(k)}_j} \cdot  \tilde H \begin{bmatrix} e^{- \pi i \lambda^{(k)}_j}& 0 \\ 0 & e^{ \pi i \lambda^{(k)}_j}   \end{bmatrix} \tilde H^T \otimes \ket{\Delta_k}\bra{\Delta_k} \otimes \ket{\psi_j}\bra{\psi_j}\\
                        &= \sum_{\Delta_k} \sum_j  e^{ \pi i \lambda^{(k)}_j} \cdot  \begin{bmatrix} \cos\left( \pi \lambda_j^{(k)} \right) & \sin\left( \pi \lambda_j^{(k)}\right)  \\ \sin\left(\pi \lambda_j^{(k)} \right)  & -\cos\left( \pi \lambda_j^{(k)} \right)   \end{bmatrix} \otimes \ket{\Delta_k}\bra{\Delta_k} \otimes \ket{\psi_j}\bra{\psi_j} \label{eqn:usignal}
        \end{align}
            In other words, $U^{(k)}_\text{signal}$ is a block-encoding of:
        \begin{align}
        \sum_{\Delta_k} \sum_j \left|\cos\left(\pi \lambda_j^{(k)}  \right)\right| \left[\pm e^{\pi i \lambda_j^{k}} \ket{\Delta_k}\ket{\psi_j}\right]  \left[ \bra{\Delta_k}\bra{\psi_j}\right]         \end{align}
            The above is a singular value decomposition of the block-encoded matrix.

        \item Choose the amplification threshold via:
            \begin{align}
                \eta_k &:= \frac{1}{2} - \frac{1}{2^k} \left( \frac{1}{2} + \frac{\alpha}{2} \right) \text{ if } k \geq 1 \\
                \eta_0 &:= \frac{\alpha}{2}
            \end{align}
            Also, let $\delta_\text{amp} < 1$ be an error threshold we will pick later. Now, let $A_{\eta\to\delta_\text{amp}}(x)$ be the polynomial from Lemma~\ref{lemma:amppoly}. Viewing $U^{(k)}_\text{signal}$ as a block-encoding of $\cos\left( \pi \lambda^{(k)}_j  \right)$, apply singular value transformation as in Lemma~\ref{lemma:svt} to $U^{(k)}_\text{signal}$ with a polynomial $\tilde p(x)$ approximating $A_{\eta\to \delta_\text{amp}}(x^2)$ to accuracy $\delta_\text{svt}$, which we also pick later. 

            Lemma~\ref{lemma:svt} applies because $A_{\eta\to \delta_\text{amp}}(x^2)$ is even. Furthermore, $U^{(k)}_\text{signal}$ only has one ancilla qubit, and has the special form where the it implements a reflection on the ancilla (\ref{eqn:usignal}). Thus, the circuit from Lemma~\ref{lemma:svt} also just has one ancilla.

            Call the resulting circuit $U^{(k)}_\text{svt}$, which implements the unitary:
\begin{align}
    U^{(k)}_\text{svt} = \sum_{\Delta_k} \sum_j \begin{bmatrix} \tilde p\left( \cos\left( \pi \lambda^{(k)}_j    \right) \right)  & \hspace{5mm}\cdot\hspace{5mm} \\ \gamma(\lambda^{(k)}_j) & \cdot  \end{bmatrix} \otimes  \left[  \ket{\Delta_k}\ket{\psi_j}\right]  \left[ \bra{\Delta_k}\bra{\psi_j}\right]     
\end{align}
            for some matrix element $\gamma(\lambda^{(k)}_j)$.  
    \end{enumerate}

Now we prove that $U^{(k)}_\text{svt}$ is an iterative phase estimator with phases. It implements the map:
        \begin{align}
            \ket{0}\ket{\Delta_k}\ket{\psi_j} \to  \left(  \tilde p\left( \cos\left(\pi \lambda^{(k)}_j \right) \right)\ket{0} + \gamma(\lambda_j^{(k)} )\ket{1}  \right)  \ket{\Delta_k}\ket{\psi_j} 
        \end{align}
    Note that $U^{(k)}_\text{svt}$ is a block-encoding with only one ancilla, and that ancilla is the output qubit of the map.

    We must show that $U^{(k)}_\text{svt}$ is close in diamond norm to a map that leaves the first qubit as $\ket{\text{bit}_k(\lambda_j)}$ whenever $\Delta_k$ encodes the $k$ least significant bits of an $n$-bit binary expansion of $\lambda_j$.

    To study this map we will proceed through the recipe above, proving statements about the expressions encountered along the way. In step 2. we defined:
    \begin{align}
            \lambda^{(k)}_j &:= 2^{n-k-1} \left(\lambda_j - \frac{\hat\Delta_k}{2^n}\right) + \phi_k,
    \end{align}
    
We discuss the relationship between $ \lambda^{(k)}_j - \phi_k $ and $\text{bit}_k(\lambda_j)$. If $k = 0$ then $\Delta_k = 0$, and due to the rounding promise we find $\lambda_j$ in regions of the form $\frac{m}{2^n} + \left[ \frac{\alpha}{2^n}, \frac{1}{2^n}  \right]$ for integers $m$.  Thus, we find $\lambda^{(0)}_j - \phi_0 $ in regions of the form $\frac{m}{2} + \left[ \frac{\alpha}{2} , \frac{1}{2} \right]$. The function $\text{bit}_k(\lambda_j)$ just indicates the parity of $m$. We can also write this as:
    \begin{align}
        \text{bit}_0(\lambda_j) = \Bigg\{ \begin{array}{ll} 0 & \text{ if } \lambda^{(0)}_j - \phi_0 \in \text{floor}(\lambda^{(0)}_j - \phi_0) +[ \frac{\alpha}{2} , \frac{1}{2} ] \\ 1 & \text{ if } \lambda^{(0)}_j - \phi_0 \in \text{floor}(\lambda^{(0)}_j - \phi_0) + [ \frac{1}{2} + \frac{\alpha}{2} , 1 ]  \end{array} 
    \end{align}

If $k > 0$ then we also can show a similar property. While for $k=0$ we used the rounding promise to guarantee that $\lambda^{(0)}_j - \phi_0 $ only falls into certain regions, for larger $k$ we simply use the fact that $\Delta_k$ has been subtracted off in the definition of $\lambda_j^{(k)}$. That means that the regions where $\text{bit}_{<k}(\lambda_j) = 1$ are no longer possible.  We find:
\begin{align}
    \text{bit}_k(\lambda_j) = \Bigg\{ \begin{array}{ll} 0 & \text{ if } \lambda^{(k)}_j - \phi_k \in \text{floor}(\lambda^{(k)}_j - \phi_k) + \left[ 0 , \frac{1}{2^k} \left( \frac{1}{2} + \frac{\alpha}{2} \right) \right] \\ 1 & \text{ if } \lambda^{(k)}_j - \phi_k \in \text{floor}(\lambda^{(k)}_j - \phi_k) +\left[ \frac{1}{2} , \frac{1}{2} +  \frac{1}{2^k} \left( \frac{1}{2} + \frac{\alpha}{2} \right) \right]  \end{array}  \text{ when } k > 0
\end{align}
Note that the claim for $k > 0$ did not make use of the rounding promise, and is true regardless of if the rounding promise holds. These regions are shown in Figure~\ref{fig:polysketch}.

    In step 3. we defined $U_\text{signal}$ which is a block-encoding of $\cos\left( \pi \lambda^{(k)}_j  \right)$. Later, this will approximately be transformed by singular value transformation via $x \to A_{\eta\to\delta_\text{amp}}(x^2)$, so we employ a trigonometric identity:
    \begin{align}
        \cos^2\left( \pi \lambda_j^{(k)}  \right) = \frac{1 + \cos\left( 2 \pi \lambda^{k}_j  \right)}{2}
    \end{align}
 Clearly this is a probability, and since cosine has period $2\pi$, the $\text{floor}(\lambda_j^{(k)} - \phi_k)$ term does not matter. We argue that this probability is either $\geq 1/2 + \eta_k$ or $\leq 1/2 - \eta_k$ depending on $\text{bit}_k(\lambda_j) $. See Figure~\ref{fig:polysketch}. The idea is that the $\phi_k$ are chosen precisely so that the troughs and peaks of $\cos^2\left(  \pi \lambda^{k}_j  \right)$ line up with the centers of the intervals corresponding to $\text{bit}_k(\lambda_j) = 0$ and $\text{bit}_k(\lambda_j) = 1$ respectively. Then, the nodes of $\cos^2\left( \pi \lambda^{k}_j  \right)$ line up with the midpoints of the gaps between the intervals. A line of slope 2 connecting a trough to a peak then forms an upper/lower bound on  $\cos^2\left( \pi \lambda^{k}_j  \right)$. If we select $ \eta_k := \frac{1}{2} - \frac{1}{2^k} \left( \frac{1}{2} + \frac{\alpha}{2} \right) \text{ if } $ and $\eta_0 := \frac{\alpha}{2}$ then these bounds show that $\cos^2\left( \pi \lambda_j^{(k)}  \right)$ is alternatingly $\leq \frac{1}{2} - \eta_k$ and $\geq \frac{1}{2} + \eta_k$.

\begin{figure}[t]
    \centering
    \includegraphics[width=0.98\textwidth]{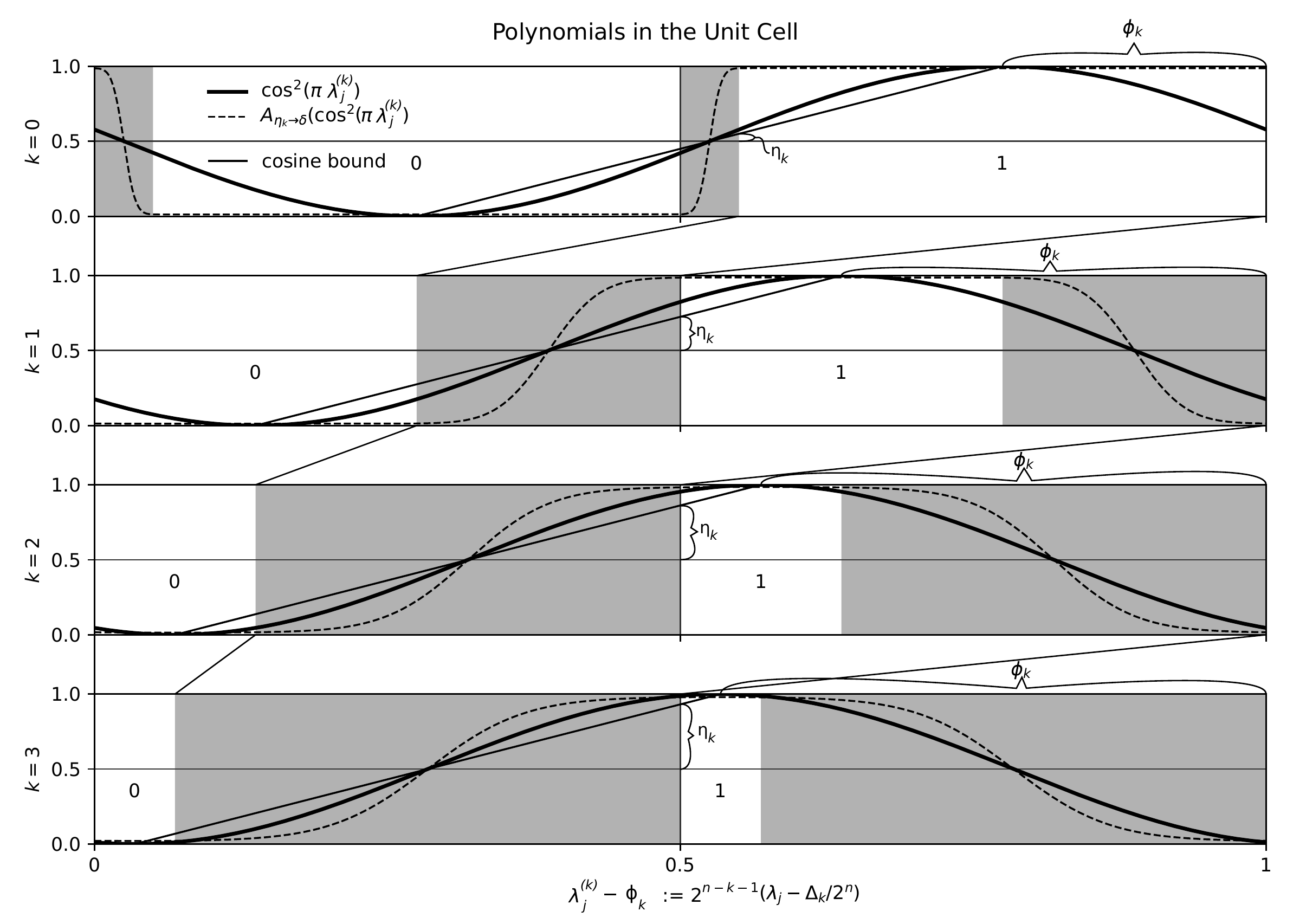}
    \caption{\label{fig:polysketch} Sketch of the probability $ \cos^2\left( \pi \lambda^{(k)}_j  \right)$ which appears in the proof of Theorem~\ref{thm:iterativephaseestimation}. \\ We are guaranteed that $\lambda^{(k)}_j$ can only appear in the un-shaded regions: for $k = 0$ the rounding promise rules out the shaded regions, and for $k \geq 1$ we have subtracted $\Delta_k/2^n$ off of $\lambda_j$, preventing regions where previous bits are 1 .\\ We can see how the probability $ \cos^2\left( \pi \lambda^{(k)}_j  \right)$ is close to 1 if $\text{bit}_k(\lambda_j) = 0$  and close to 0 if $\text{bit}_k(\lambda_j) = 1$. To make this claim precise, we simply fit a line with slope $2$ to the points where the probability intersects $1/2$, and see that this line alternatingly gives upper or lower bounds on the probability. So if $\eta_k/2$ is equal to half the distance between allowed intervals, then the probability is either $\geq 1/2 +\eta_k$ or $\leq 1/2-\eta_k$ in the regions where $\lambda^{(k)}_j$ can appear.\\
    The relationship of this figure to Figure~\ref{fig:projectionpolys} is fairly simple: the probability as a function of $\lambda_j$ can be obtained by just tiling the `unit cell' shown in this figure $2^{n-k-1}$ times. This also makes the ratio $2^{n-k-1}$ between $\lambda_j$ and $\lambda^{(k)}_j$ intuitive.
   }
\end{figure}

Then we have:
    \begin{align}
        \cos^2(\pi \lambda_j^{(k)}) \text{ is } \Bigg\{ \begin{array}{ll} \leq \frac{1}{2} + \eta_k & \text{ if } \text{bit}_k(\lambda_j) = 0 \\[2mm] \geq \frac{1}{2} - \eta_k & \text{ if } \text{bit}_k(\lambda_j) = 1   \end{array}
    \end{align}
By Lemma~\ref{lemma:amppoly}:
    \begin{align}
        A_{\eta_k\to\delta_\text{amp}} \left(\cos^2(\pi \lambda_j^{(k)}) \right) \text{ is } \Bigg\{ \begin{array}{ll} \leq \delta_\text{amp} & \text{ if } \text{bit}_k(\lambda_j) = 0 \\[2mm] \geq 1-\delta_\text{amp} & \text{ if } \text{bit}_k(\lambda_j) = 1   \end{array}
    \end{align}

And finally, since $\tilde p(x)$ approximates $A_{\eta_k\to\delta_\text{amp}}(x^2)$ to accuracy $\delta_\text{svt}$:
    \begin{align}
        \tilde p\left(\cos(\pi \lambda_j^{(k)}) \right) \text{ is } \Bigg\{ \begin{array}{ll} \leq \delta_\text{amp}+\delta_\text{svt} & \text{ if } \text{bit}_k(\lambda_j) = 0 \\[2mm] \geq 1-\delta_\text{amp}-\delta_\text{svt} & \text{ if } \text{bit}_k(\lambda_j) = 1   \end{array}
    \end{align}

The circuit for $U_\text{svt}$ is completely unitary, so the resulting state is normalized. Therefore:
\begin{align}
    \left| \tilde p \left(\cos\left(\pi \lambda^{(k)}_j \right) \right)  \right|^2  + \left| \gamma(\lambda^{(k)}_j )\right|^2 = 1
\end{align}

Using this fact we can reason that if $\tilde p \left(\cos\left(\pi \lambda^{(k)}_j  \right) \right) \leq \delta_\text{amp} + \delta_\text{svt}$ then $|\gamma(\lambda_j^{(k)})|^2 \geq 1 - (\delta_\text{amp} + \delta_\text{svt})$.

Similarly, if $\tilde p \left(\cos\left( \pi \lambda_j^{(k)}  \right) \right) \geq 1 -  \delta_\text{amp} - \delta_\text{svt}$, then:
\begin{align}
    |\gamma(\lambda_j^{(k)})|^2 \leq 1 - (1 - (\delta_\text{amp} + \delta_\text{svt}))^2 \leq 2(\delta_\text{amp} + \delta_\text{svt}) - (\delta_\text{amp} + \delta_\text{svt})^2
\end{align}

Now that we have bounds on the amplitudes of the output state, we can bound its distance to $\ket{\text{bit}_k(\lambda_j)}$ for a favorable choice of $e^{i\varphi_j}$. Say $\text{bit}_k(\lambda_j) = 0$. Then we select $\varphi_j = 0$, so that:
    \begin{align}
        &\left| \left(  \tilde p\left( \cos\left(   \pi \lambda^{(k)}_j   \right) \right)\ket{0} + \gamma(\lambda^{(k)}_j)\ket{1}  \right)  - e^{i\varphi_j} \ket{\text{bit}_k(\lambda_j)} \right| \\
        \leq& \sqrt{  \left|\tilde p \left(\cos\left(\pi \lambda_j^{(k)} \right) \right)- 1\right|^2  +  \left|\gamma(\lambda_\Delta) \right|^2 }\\
        \leq& \sqrt{  (\delta_\text{amp} + \delta_\text{svt})^2 + 2(\delta_\text{amp} + \delta_\text{svt}) - (\delta_\text{amp} + \delta_\text{svt})^2}\\
        \leq& \sqrt{  2(\delta_\text{amp} + \delta_\text{svt}) }
    \end{align}

Otherwise, if $\text{bit}_k(\lambda_j) = 1$, then we define $\varphi_j$ by $ \gamma(\lambda_j^{(k)}) = e^{i\varphi_j}|\gamma(\lambda_j^{(k)})|$. That way:

    \begin{align}
        &\left| \left(  \tilde p\left( \cos\left(   \pi \lambda_j^{(k)}  \right) \right)\ket{0} + \gamma(\lambda_j^{(k)})\ket{1}  \right)  - e^{i\varphi_j} \ket{\text{bit}_k(\lambda_j)} \right| \\
        \leq& \sqrt{  \left|\tilde p \left(\cos\left(\pi \lambda_j^{(k)} \right) \right)\right|^2  +  \left|e^{i\varphi_j}( | \gamma(\lambda_j^{(k)})| - 1) \right|^2 }\\
        \leq& \sqrt{  \left|\delta_\text{amp} + \delta_\text{svt}\right|^2  +  \left|\delta_\text{amp} + \delta_\text{svt}\right|^2 }\\
        \leq& \sqrt{  2(\delta_\text{amp} + \delta_\text{svt}) }
    \end{align}
So either way, the output state is within $\sqrt{  2(\delta_\text{amp} + \delta_\text{svt}) }$ in spectral norm of that of the ideal state. Thus, the unitary map $U_\text{svt}$ is close in spectral norm to some ideal unitary. Invoking Lemma~\ref{lemma:spectraltodiamond}, the distance in diamond norm is at most:
    \begin{align}
       2 \sqrt{  2(\delta_\text{amp} + \delta_\text{svt}) } \leq \delta \label{eqn:errorbound}
    \end{align}
The inequality above holds if we select, for any $m > 0$:
\begin{align}
    \delta_\text{amp} := (1 - 10^{-m}) \cdot \frac{\delta^2}{8} \hspace{1cm} \delta_\text{svt} := 10^{-m}  \cdot \frac{\delta^2}{8}
\end{align}
This solution to the inequality relies on the fact that only $\delta_\text{amp}$ actually enters the query complexity, so if classical resources are cheap but query complexity is expensive then we can make the classical computer do as much work as possible by making $m$ larger.

  Finally, we compute the query complexity. An invocation of $e^{2\pi i \hat\lambda^{(k)}_j}$ requires $2^{n-k-1}$ invocations of $U = e^{2\pi i \hat\lambda}$. By Lemma~\ref{lemma:svt}, the number of queries made by the unitary $U_\text{svt}$ to $U_\text{signal}$ is the degree of the polynomial. By Lemma~\ref{lemma:amppoly}, the polynomial $A_{\eta\to\delta_\text{amp}}(x^2)$ has degree
  $2 \cdot M_{\eta\to\delta}$, so the query complexity is:
 \begin{align} 2^{n-k-1} \cdot  2 \cdot  M_{\eta\to\delta_\text{amp}}\end{align}
\end{proof}


A really nice feature of the coherent iterative phase estimator we present is that it produces no garbage qubits. All singular value transformation is performed on the final output qubit. It does still produce extra phase shifts between the eigenstates, which in some applications may still need to be be uncomputed. However, in applications where phase differences between eigenstates do not matter, like thermal state preparation, we expect that this uncomputation step can be skipped.

To finish the discussion of coherent iterative phase estimation, we stitch the iterative estimator we just defined into a regular phase estimator. In doing so, we also remark on what happens when no rounding promise is present. A summary of the construction is presented in Figure~\ref{fig:circuit}.


\begin{corollary} \label{cor:improvedphaseestimation} \textbf{Improved phase estimation.} The coherent iterative phase estimator from Theorem~\ref{thm:iterativephaseestimation} can be combined with Lemma~\ref{lemma:stitchingiterative} to make a phase estimator with phases with query complexity at most:
    \begin{align}
        O\left(  2^{n} \alpha^{-1} \log( \delta^{-1} )     \right)
    \end{align}
    assuming $\alpha$ is bounded away from $1$ by a constant.

Furthermore, if no rounding promise is given, then the estimator $\delta$-approximates in diamond norm a map:
    \begin{align}
        \ket{0^n}\ket{\psi_j} \to \left( \xi \ket{\text{floor}(2^n\lambda_j)} + \zeta \ket{\lambda'_j}  \right) \ket{\psi_j}
    \end{align}
    for some complex amplitudes $\xi,\zeta$ and $\lambda'_j = \text{floor}(2^n\lambda_j)-1 \text{ mod } 2^n$ is an erroneous estimate. The performance is the same, except that $0 < \alpha < 1$ can be any constant.

\end{corollary}
\begin{proof} Write  $\eta_k := \frac{1}{2} - \frac{1}{2^k} \left( \frac{1}{2} + \frac{\alpha}{2} \right) $ if  $ k \geq 1$ and $\eta_0 := \frac{\alpha}{2}$ if $k=0$, and, following Lemma~\ref{lemma:stitchingiterative}, demand an accuracy of $\delta_{\text{amp},k} := (1-10^{-m}) ( \delta 2^{-k-1} )^2  /8 $ for the $k$'th bit. Recall from Lemma~\ref{lemma:amppoly} that $M_{\eta_k\to\delta_\text{amp}} \in O\left(  \eta_k^{-1}\log(\delta_\text{amp}^{-1})   \right)$. Then, the overall query complexity is:
    \begin{align}
        \sum_{k=0}^{n-1} 2^{n-k} \cdot M_{\eta_k\to\delta_{\text{amp},k}}  &\in O\left( \sum_{k=0}^{n-1} 2^{n-k}  \eta_k^{-1}  \log( \delta_{\text{amp},k}^{-1} )   \right)\\
        &= O\left( \sum_{k=0}^{n-1} 2^{n-k}  \eta_k^{-1}  \log(  2^{k+1} \delta^{-1} )   \right)
    \end{align}
    When $ k = 0$ we have $\eta_0 = \frac{\alpha}{2} $, and when $k > 0$ we have $\eta_k > \frac{1-\alpha}{4}$ which is bounded from below by a constant. We can use this to split the sum:
    \begin{align}
         &\leq  O\left(  2^{n-0} \eta_0^{-1} \log( 2^{0+1}  \delta^{-1} ) + \sum_{k=1}^{n-1} 2^{n-k} \eta_k^{-1} \log( 2^{k+1} \delta^{-1} ) \right)\\
        &\leq O\left( 2^n \alpha^{-1} \log(\delta^{-1}) + \sum_{k=1}^{n-1} 2^{k} (k + \log(\delta^{-1})) \right)\\
        &\leq O\left( 2^n \alpha^{-1} \log(\delta^{-1}) + 2(2^n - n - 1) + (2^n - 2) \log(\delta^{-1}) \right)\\
        &\leq O\left(  2^{n} \alpha^{-1} \log( \delta^{-1} )     \right)
    \end{align}

    This completes the runtime analysis. Now we turn to the case when no rounding promise is present. Notice that when $k \geq 1$, the regions where $\lambda_j$ is assumed not to appear are guaranteed by the previous estimators \emph{definitely} outputting $1$ for previous bits regardless of the promise. Thus, the only bit that can be wrong is the first bit. When an eigenvalue $\lambda_j$ falls into a region disallowed by the rounding promise, the first bit will be some superposition $\xi\ket{0} + \zeta\ket{1}$.
    
    Flipping the final bit of an estimate in general results in an error of $\pm 1$. However, recall that when estimating future bits the value of $\Delta_k/2^n$ is subtracted off from $\lambda_j$. That means that when we erroneously measure a $1$, the rest of the algorithm proceeds to measure $\lambda_j - 1/2^n$ instead. As a result, if the first bit is wrong, then the algorithm will output $\text{floor}(2^n \lambda_j) - 1$ instead. Since the algorithm is periodic in $\lambda_j$ with period $1$, the output will be $2^n-1$ if an error occurs when $\text{floor}(2^n \lambda_j) = 0$.

\end{proof}

    Notice that if we had allocated the error evenly between the $n$ steps, then we would have incurred an extra $ O(2^n  \log n)$ term in the above. Spreading the error via a geometric series avoids this, and we find that we obtain better constant factors with this choice as well. This is because the $k=0$ term dominates, so we want to make $\delta_{\text{amp},k}$ as large as possible.


\begin{figure}[h]
    \centering
    \includegraphics[width=\textwidth]{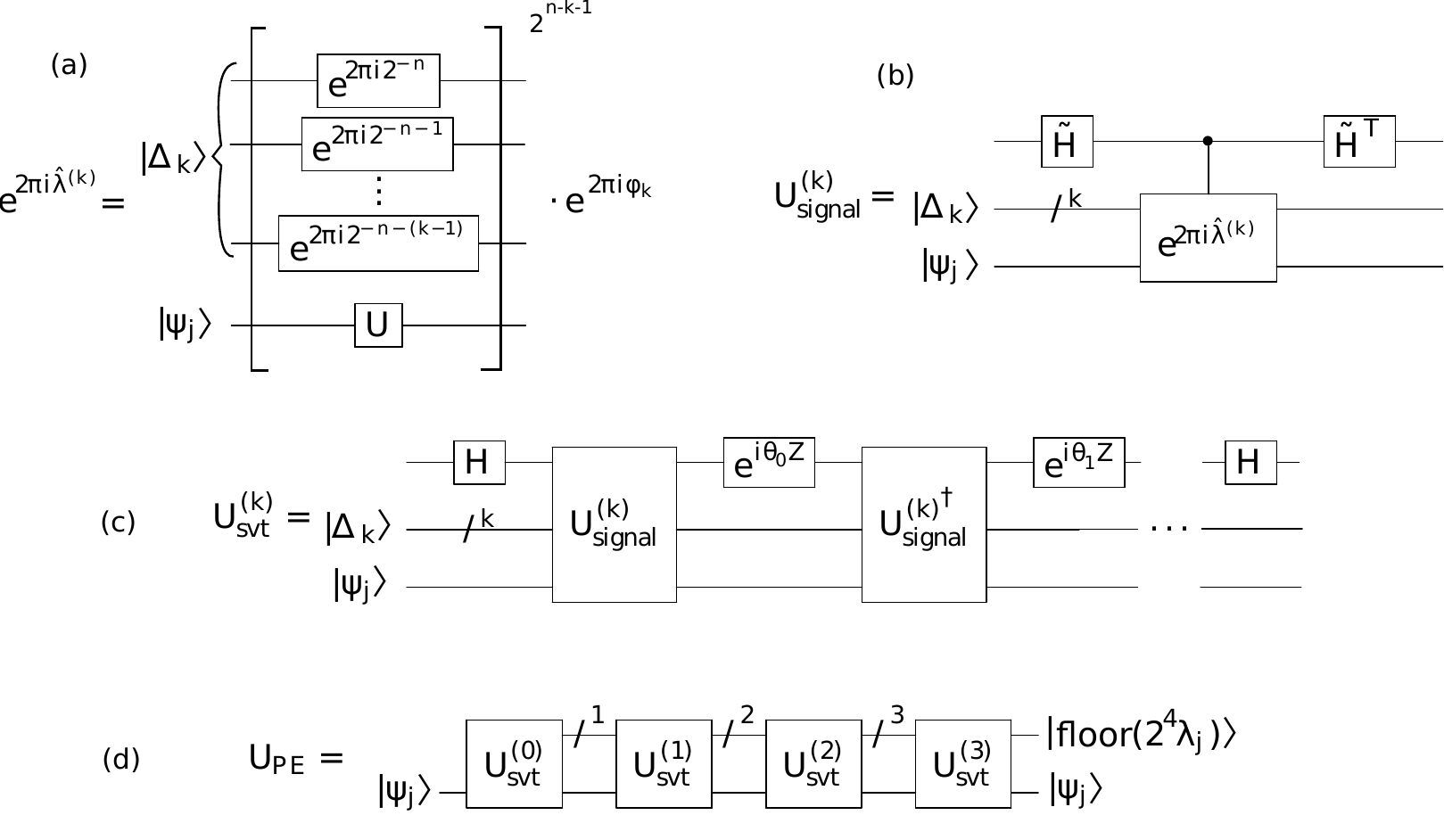}
    
    \caption{ \label{fig:circuit} Circuit diagram for the algorithm in Corollary~\ref{cor:improvedphaseestimation}. (a) and (b) are depictions of (\ref{eqn:lambdakunitary}) and (\ref{eqn:usignal_def}) respectively. (c) depicts the singular value transformation circuit obtained from Lemma~\ref{lemma:svt} which interperses alternating applications of $U^{(k)}_\text{signal}$ and $U_\text{signal}^{(k)\dagger}$  with phase rotations $e^{i \theta_j Z}$, where the angles $\theta_j$ encode the polynomial approximating $A_{\eta\to\delta_\text{amp}}(x^2)$. (d) depicts the final circuit assembled via Lemma~\ref{lemma:stitchingiterative}, showing how each iterative estimator's output becomes part of the $\ket{\Delta_k}$ register for the next, and how that register finally becomes the output.  } 
\end{figure}

\subsection{Coherent Iterative Energy Estimation}   \label{sec:energyestimation}

While for phase estimation we are given access to $\sum_j e^{2\pi i \lambda_j} \ket{\psi_j}\bra{\psi_j}$, for energy estimation we have a block-encoding of $\sum_j \lambda_j \ket{\psi_j}\bra{\psi_j}$. For phase estimation we could relatively easily synthesize a cosine in the eigenvalues, just by taking a linear combination with the identity. But for energy estimation we must build the cosine directly.

To do so we leverage a tool employed by \cite{1806.01838} to perform Hamiltonian simulation. The Jacobi-Anger expansion yields highly efficient polynomial approximations to $\sin(tx)$ and $\cos(tx)$. To perform Hamiltonian simulation one then takes the linear combination $\cos(tx) + i\sin(tx)$. However, we only need the cosine component.


\begin{lemma} \textbf{Jacobi-Anger expansion.} \label{lemma:jacobianger} For any $t \in \mathbb{R}^+$, and any $\eps \in (0,1/e)$, let:
    \begin{align}
        r(t',\eps') &:= \text{ the solution to } \eps' = \left(\frac{t'}{r}\right)^r \text{ such that } r \in (t', \infty),\label{eq:rdef}\\
        R &:= \left\lfloor r\left(\frac{et}{2} , \frac{5}{4}\eps\right)/2 \right\rfloor \\
        J_k(t) &:= \text{the } k\text{'th Bessel function of the first kind}\\
        T_k(t) &:= \text{the } k\text{'th Chebyshev polynomial of the first kind}\\
        p_{\text{cos},t}(x) &:= J_0(t) + 2\sum_{k=1}^R (-1)^k J_{2k}(t) T_{2k}(x) 
    \end{align}
    Then $p_{\text{cos},t}(x)$ is an even polynomial of degree $2R$ such that for all $x \in [-1,1]$:
        \begin{align}
            \left| \cos(tx) - p_{\text{cos},t}(x) \right| \leq \eps.
        \end{align}
    Furthermore:
        \begin{align}
            r(t',\eps') \in \Theta\left(t' + \frac{\log \eps'^{-1}}{\log(\log(\eps'^{-1}))}\right)
        \end{align}
\end{lemma}
\begin{proof} These results are shown in Lemma~57 and Lemma~59 of \cite{1806.01838}, outlined in their section 5.1.
\end{proof}


Again, computation of the degree of the Jacobi-Anger expansion is a bit complicated, but the complexity is worth it due to the method's high performance. Our approach is to first synthesize $\cos\left(\pi \lambda^{(k)}_j\right)$ to obtain a signal that oscillates to indicate $\text{bit}_k(\lambda_j)$, and then apply $A_{\eta\to\delta}(x^2)$ to amplify the signal to $0$ or $1$, as shown in Figure~\ref{fig:projectionpolys}.

Actually, one might observe that synthesizing $\cos\left(\pi \lambda^{(k)}_j\right)$ first is not necessary to make polynomials that look like those in Figure~\ref{fig:projectionpolys}. Instead one can take an approach similar to the one used for making rectangle functions in \cite{1707.05391}: simply shift, scale and add several amplifying polynomials $A_{\eta\to\delta}(x)$ to make the desired shape. None of these operations affect the degree, so this approach also yields the same asymptotic scaling $O(2^n \alpha^{-1} \log(\delta^{-1}))$ as phase estimation. We will see in Corollary~\ref{cor:improvedenergyestimation} that the method using the Jacobi-Anger expansion actually achieves the worse scaling of $O( \alpha^{-1} \log(\delta^{-1}) (2^n + \log(\alpha^{-1})) )$. 

The reason why we present the approach using the Jacobi-Anger expansion, despite it having worse asymptotic scaling, is that in the regime of interest ($n\approx 10, \alpha \approx 2^{-10}$) we numerically find that the Jacobi-Anger expansion actually performs better. There may be a regime where it is better to remove the Jacobi-Anger expansion from the construction, in which case the algorithm is easily adapted.

The rest of the construction of Theorem~\ref{thm:iterativeenergyestimation} strongly resembles coherent iterative phase estimation, so much so that we can re-use parts of the proof of Theorem~\ref{thm:iterativephaseestimation}. One further difference is that this estimator now has garbage, because we have no guarantee that the block-encoding of the Hamiltonian only has one ancilla.

\begin{theorem} \label{thm:iterativeenergyestimation} \textbf{Coherent Iterative Energy Estimation.} Say the block-encoding of $H$ requires $a$ ancillae, that is, $U_H$ acts on $\mathbb{C}^{2^a} \otimes \mathcal{H}$. Then there is an iterative energy estimator with phases and $a+n+3$ qubits of garbage with query complexity:
\begin{align}
    4 \cdot M_{  (1-10^{-m_\text{cos}})\eta_k \to \delta_\text{amp} }  \cdot \left\lfloor  r\left( \frac{e}{2} \pi  2^{n-k} , \frac{5}{4} \frac{\eta_k}{2}  10^{-m_\text{cos}}\right) \right\rfloor
\end{align}
    where $M_{\eta\to\delta}$ and $\eta_k$ are as in Lemma~\ref{lemma:amppoly} and  $\delta_\text{amp}$ can be chosen to be $(1 - 10^{-m_\text{svt}})\delta^2 / 8$ for any $ m_\text{svt} > 0$, and we can choose any $m_\text{cos} > 0$.
\end{theorem}
\begin{proof}
        We construct the estimator as follows:
    \begin{enumerate}
     \item Let $W_k$ be a hermitian matrix on $k$ qubits defined by:
            \begin{align}
                W_k := 2 \sum_{\Delta_k =0}^{2^k - 1} \frac{\Delta_k}{2^{n}}  \ket{\Delta_k}\bra{\Delta_k}
            \end{align}
            $W_k$ has a block-encoding which can be constructed as follows: First, prepare $\ket{+^{n-1}}$:
            \begin{align}
                \ket{\Delta_k} \to \ket{\Delta_k}\ket{+^{n-1}} = \frac{1}{\sqrt{2^{n-1}}} \sum_{x=0}^{2^{n-1}-1} \ket{\Delta_k} \ket{x}
            \end{align}
            Next, observe that $\Delta_k \leq 2^{k} -1 \leq 2^{n-1} - 1$. Into an ancilla register compute $\ket{x < \Delta_k}$, and uncompute any garbage necessary to do so.
            \begin{align}
           \to \frac{1}{\sqrt{2^{n-1}}} \sum_{x = 0}^{2^{n-1}-1} \ket{\Delta_k}\ket{x}   \ket{x < \Delta_k}
        \end{align}
            Then postselect that the final register is in the $\ket{1}$ state:
            \begin{align}
                \to \frac{1}{\sqrt{2^{n-1}}} \ket{\Delta_k} \sum_{x = 0}^{\Delta_k -1 } \ket{x}
            \end{align}
                Finally, postselect that the $x$ register is in the $\ket{+^{n-1}}$ state:
            \begin{align}
                \to \ket{\Delta_k}\sum_{x = 0}^{\Delta_k -1 }  \frac{1}{2^{n-1}}     =  \frac{\Delta_k}{2^{n-1}} \ket{\Delta_k} 
            \end{align}
            This process makes use of $n-1$ ancillae initialized and postselected to $\ket{+}$, and one more ancilla that is postselected to $\ket{1}$.

        \item Use linear combinations of unitaries to construct a block-encoding of:
            \begin{align}
                H^{(k)} := \frac{1}{2}\cdot I \otimes H - \frac{1}{4}\cdot W_k \otimes I +  \frac{1}{4}\cdot (4 \phi_k 2^{k-n}) \cdot I \otimes I
            \end{align}
            where the $\phi_k$ are selected just as in as in Theorem~\ref{thm:iterativephaseestimation}. Observe that $\phi_k < 1/2$, so therefore $4 \phi_k 2^{k-n}$ is a probability which can be block-encoded. Since the block-encoding of $H$ has $a$ ancillae, and that of $W_k$ has $n$ ancillae, and the three terms in the linear combination need two control qubits, the block-encoding of $H^{(k)}$ has $a+n+2$ ancillae.

        \item Use Lemma~\ref{lemma:jacobianger} to construct a polynomial approximation $p_{\text{cos}, \pi 2^{n-k} }(x)$ of $\cos( \pi 2^{n-k}   x)$ to accuracy $2\delta_\text{cos}$, to be picked later. 

            Use Lemma~\ref{lemma:amppoly} to construct a polynomial $A_{(\eta-\delta_\text{cos})\to\delta_\text{amp}}$. We will pick $\delta_\text{amp}$ later and select $\eta_k$ just as in Theorem~\ref{thm:iterativephaseestimation}.

            Finally, use Lemma~\ref{lemma:svt} to construct $U_\text{svt}$, a block-encoding of $\tilde p(H^{(k)})$ which approximates the even polynomial:
            \begin{align}
                \tilde p(x) \approx A_{(\eta_k - \delta_\text{cos})\to\delta_\text{amp}}\left( p^2_{\text{cos},\pi 2^{n-k}}\left( x \right)  \right)
            \end{align}
           To perform singular value transformation we needed one extra ancilla, so $U_\text{svt}$ has $a+n+3$ ancillae - these are the garbage output of this map.
        \item Use the modified Toffoli gate $I \otimes \ket{0}\bra{0} + X \otimes (I - \ket{0}\bra{0}) $ to conditionally flip the output qubit.
    \end{enumerate}

We rewrite $H^{(k)}$ in terms of its eigendecomposition:
    \begin{align}
        H^{(k)} &= \sum_j \sum_{\Delta_k} \left[ \frac{\lambda_j}{2} - \frac{1}{2} \frac{\Delta_k}{2^n} +2^{k-n} \phi_k \right] \ket{\Delta_k}\bra{\Delta_k} \otimes \ket{\psi_j}\bra{\psi_j} \\
        &= \sum_j \sum_{\Delta_k}  2^{k-n} \lambda^{(k)}_j \ket{\Delta_k}\bra{\Delta_k} \otimes \ket{\psi_j}\bra{\psi_j}  
    \end{align}
    where we defined $\lambda_j^{(k)} := 2^{n-k-1}\left( \lambda_j - \frac{\Delta_k}{2^n}\right) + \phi_k$ just like in Theorem~\ref{thm:iterativephaseestimation}. This protocol implements a map:
    \begin{align}
        \ket{0}\ket{0...0}\ket{\Delta_k}\ket{\psi_j} \to \left( \tilde p\left(  2^{k-n}  \lambda^{k}_j\right) \ket{0} \ket{\text{gar}_{0,j}} + \gamma(\lambda^{k}_j\ket{1}  \ket{\text{gar}_{1,j}} \right)\ket{\Delta_k}\ket{\psi_j}
    \end{align}
    Here $\gamma(\lambda_j^{(k)})$ is defined such that all the other amplitudes for the failed branches of the block-encoding are absorbed into the normalized state $\ket{\text{gar}_{1,j}}$.

    We see that $\tilde p\left(  2^{k-n}  \lambda^{k}_j\right)$ approximates
    \begin{align}
        \tilde p(\lambda_\Delta/2) \approx A_{(\eta_k - \delta_\text{cos})\to\delta_\text{amp}}\left( \cos^2\left( \pi  \lambda^{(k)}_j \right)  \right) 
    \end{align}
    which is the same expression encountered in Theorem~\ref{thm:iterativephaseestimation}, with the same definition of $\lambda^{(k)}_j$, up to a minor shift on $\eta_k$. Therefore, we can follow the same reasoning as in Theorem~\ref{thm:iterativephaseestimation} up to two minor differences, and arrive at the exact same conclusion. Namely, if we select some $m_\text{svt} > 0$ and then let:
    \begin{align}
        \delta_\text{amp} := (1-10^{-m_\text{svt}}) \frac{\delta^2}{8}, \hspace{1cm} \delta_\text{svt} := 10^{-m_\text{svt}} \frac{\delta^2}{8},
    \end{align}
Then the unitary channel we implement is at most $\delta$-far in diamond norm to a channel that implements the map:
    \begin{align}
        \ket{0}\ket{0...0}\ket{\Delta_k}\ket{\psi_j} \to e^{i\varphi_j} \ket{\text{bit}_k(\lambda_j)}  \ket{\text{gar}_{\lambda_j}} \ket{\Delta_k}\ket{\psi_j}
    \end{align}
    whenever $\Delta_k$ encodes the last $k$ bits of $\lambda_j$.

Since the argument in Theorem~\ref{thm:iterativephaseestimation} is lengthy and the modifications are extremely minor, we will not repeat the argument here. Instead we will articulate the two things that change.

First, Theorem~\ref{thm:iterativephaseestimation} gives an estimator with phases and no garbage, whereas here we also have garbage. The garbage just tags along for the entire calculation, and when we come to selecting $\varphi_j$ depending on $\lambda_j$ we can also select $ \ket{\text{gar}_{\lambda_j}} = \ket{\text{gar}_{0/1,j}}$ depending on $\text{bit}_{k}(\lambda_j)$. 

Second, we incur an error of $\delta_\text{cos}$ in the approximation of $\cos(\pi 2^{n-k} x)$ with $p_{\text{cos},\pi 2^{n-k}}(x)$. Since we show that $\cos^2( \pi 2^{n-k} x)$ is bounded away from $\frac{1}{2}$ by $\pm \eta$ we therefore have that $p_{\text{cos}, \pi 2^{n-k}}(x)^2$ is bounded away from $\frac{1}{2}$ by $\pm\left(\eta - \frac{2\delta_\text{cos}}{2}\right)$.  The amplifying polynomial then proceeds to amplify $(\eta - \delta_\text{cos}) \to \delta_\text{amp}$ appropriately, so the calculation proceeds the same. We just need to ensure that $\eta - \delta_\text{cos} > 0$, so we select
\begin{align}
    \delta_\text{cos} := 10^{-m_\text{cos}} \cdot \eta  
\end{align}
for some $m_\text{cos} > 0$. This completes the accuracy analysis.

Finally, we analyze the query complexity. The block-encoding of $H^{(k)}$ makes one query to $U_H$, so by Lemma~\ref{lemma:svt} the query complexity is exactly the degree of $\tilde p(x)$ which is the degree of $A_{( \eta  - \delta_\text{cos}) \to \delta_\text{amp}}\left( p^2_{\text{cos},\pi 2^{n-k}}(x) \right)$. From Lemma~\ref{lemma:jacobianger} and Lemma~\ref{lemma:amppoly}, the degree is:
\begin{align}
    M_{  (\eta - \delta_\text{cos}) \to \delta_\text{amp} }  \cdot 2 \cdot 2\left\lfloor  r\left( \frac{e}{2} \pi  2^{n-k} , \frac{5}{4} \frac{\delta_\text{cos}}{2} \right) \right\rfloor
\end{align}
Substituting the definitions for $\eta, \delta_\text{amp}$ and $\delta_\text{cos}$ yields the final runtime. As with Theorem~\ref{thm:iterativephaseestimation}, $\delta_\text{amp}$ can be made larger by increasing $m_\text{amp}$. By increasing $m_\text{cos}$ we can decrease $\delta_\text{cos}$, which decreases the degree of $A_{(\eta - \delta_\text{cos})\to\delta_\text{amp}}(x)$ while increasing the degree of $p_{\text{cos},\pi 2^{n-k}}(x)$. The Jacobi-Anger expansion deals with error more efficiently than the amplifying polynomial, so in practice $m_\text{cos}$ should be quite large.
\end{proof}


As with phase estimation, we also pack the coherent iterative energy estimator into a regular energy estimator using Lemma~\ref{lemma:stitchingiterative}. This time, since the iterative estimator produces garbage it makes sense to uncompute the garbage using Lemma~\ref{lemma:iterativegarbageremove}.


\begin{corollary} \label{cor:improvedenergyestimation} \textbf{Improved Energy Estimation.} The iterative energy estimator with phases and garbage from Theorem~\ref{thm:iterativeenergyestimation} can be combined with Lemma~\ref{lemma:iterativegarbageremove} and Lemma~\ref{lemma:stitchingiterative} to make an energy estimator without phases and without garbage with query complexity bounded by:
\begin{align}
    O\left( \alpha^{-1}  \log(\delta^{-1})\left( 2^{n} + \log\left( \alpha^{-1} \right)   \right)  \right)
\end{align}
assuming that $\alpha$ is bounded away from 1 by a constant.

    Furthermore, even when there is no rounding promise, there exists an algorithm that, given an eigenstate $\ket{\psi_j}$ of the Hamiltonian $\lambda_j$,  for any $\delta > 0$ performs a transformation $\delta$-close in diamond norm to the map:
    \begin{align}
        \ket{\psi_j}\bra{\psi_j} \to  \left( p  \ket{\text{floor}(2^n\lambda_j)} \bra{\text{floor}(2^n\lambda_j)} + (1-p) \ket{\lambda'_j}\bra{\lambda'_j} \right) \ket{\psi_j}\bra{\psi_j} 
    \end{align}
    where $p$ is some probability and $\lambda'_j = \text{floor}(2^n \lambda_j) - 1$ mod $2^n$ is an erroneous estimate. Just as in Corollary~\ref{cor:improvedphaseestimation}, the performance is the same except that $0 < \alpha < 1$ can be any constant.

\end{corollary}
\begin{proof} As with Corollary~\ref{cor:improvedphaseestimation}, we write $\eta_k$ to make the $k$ dependence explicit and demand accuracy $\delta_{\text{amp},k} = (1-10^{-m_\text{amp}}) (\delta 2^{-k-1})^2 / 8$ for the $k$'th bit. From Lemma~\ref{lemma:jacobianger} we obtain an asymptotic upper bound $r(t,\eps) \in O\left( t + \log(\eps^{-1})  \right)$. Again, recall from Lemma~\ref{lemma:amppoly} that $M_{\eta_k\to\delta_\text{amp}} \in O\left(  \eta_k^{-1}\log(\delta_\text{amp}^{-1})   \right)$. The asymptotic query complexity of the iterative energy estimator from Theorem~\ref{thm:iterativeenergyestimation} is then bounded by:
    \begin{align}
        &O\left(  (\eta_k - \delta_\text{cos})^{-1}  \log(\delta^{-1}_{\text{amp},k})  \cdot (    2^{n-k} + \log( \delta_{\text{cos},k}  )    )   \right)\\
        \leq& O\left( (1- 10^{-m_\text{cos}})^{-1}  \eta_k^{-1}  \log( (1 - 10^{-m_\text{amp}})^{-1}  2^{2(k+1)} \delta^{-2}  8 )  \cdot  (    2^{n-k} + \log( 10^{m_\text{cos}}  \eta_k^{-1}  )    )   \right)\\
        \leq& O\left( \eta_k^{-1}  \log(  2^{k+1} \delta^{-1}  )  \cdot  (    2^{n-k} + \log(  \eta_k^{-1}  )    )   \right)    \end{align}
Next we invoke Lemma~\ref{lemma:iterativegarbageremove} to remove the phases and the garbage, doubling the query complexity. We do this before invoking Lemma~\ref{lemma:stitchingiterative}, because Lemma~\ref{lemma:stitchingiterative} involves blowing up the number of garbage registers by a factor of $n$. While we could also invoke Lemma~\ref{lemma:stitchingiterative} and then invoke Lemma~\ref{lemma:garbageremove} to obtain a map without garbage, this would involve many garbage registers sitting around waiting to be uncomputed for a long time.
If we invoke Lemma~\ref{lemma:iterativegarbageremove} first we get rid of the garbage immediately.

  Finally, we invoke Lemma~\ref{lemma:stitchingiterative} to turn our iterative energy estimator without garbage and phases into a regular energy estimator without garbage and phases. As in Corollary~\ref{cor:improvedphaseestimation}, we observe that $\eta_0 = \alpha/2$ and for $k>0$ we have $\eta_k$ bounded from below by a constant. Then, the total query complexity is:
    \begin{align}
        & O\left(  \eta_0^{-1} \log(2^{0+1} \delta^{-1} )  \left( 2^{n-0} + \log\left( \eta_0^{-1} \right)   \right) + \sum_{k=1}^{n-1} \eta_k^{-1} \log(2^{k+1} \delta^{-1} )  \left( 2^{n-k} + \log\left( \eta_k^{-1} \right)   \right) \right) \\
        \leq & O\left(  \alpha^{-1} \log(\delta^{-1})  \left( 2^{n} + \log\left( \alpha^{-1} \right)   \right) + \sum_{k=1}^{n-1}  (k +   \log(\delta^{-1} ))  2^{n-k}  \right) \\
        \leq & O\left( \alpha^{-1} \log(\delta^{-1})  \left( 2^{n} + \log\left( \alpha^{-1} \right)   \right) + 2(2^n -n -1) + (2^n - 2)\log(\delta^{-1})  \right)\\
        \leq & O\left( \alpha^{-1}  \log(\delta^{-1})  \left( 2^{n} + \log\left( \alpha^{-1} \right)   \right)  \right)
    \end{align}

    Next, we show that it is possible to implement a map that, given an eigenstate $\ket{\phi_j}$, measures an estimate that is either $\text{floor}(2^n \lambda_j)$ or $\text{floor}(2^n \lambda_j) - 1$ mod $2^n$ with some probability. Note that this is not the same algorithm as above. Just as with phase estimation, it is only the first bit that actually needs the rounding promise, and all other bits are guaranteed to be deterministic. 

    The first bit performs a map of the form:
    \begin{align}
        \ket{0^n}\ket{0...0}\ket{\psi_j} \to \left( \sqrt{p}\ket{0}\ket{\text{gar}_{0,j}} + \sqrt{1-p}\ket{1}\ket{\text{gar}_{1,j}}  \right) \ket{\psi_j}
    \end{align}
    We immediately see that Lemma~\ref{lemma:iterativegarbageremove} cannot be used to perform uncomputation here, because uncomputation only works when $p = 1$ or $p = 0$. Instead, we simply measure the output register containing $\ket{0}$ or $\ket{1}$ and discard the garbage. This would damage any superposition over the $\ket{\psi_j}$, which is why this algorithm only works if the input is an eigenstate.

    We then use iterative estimators for the remaining bits to compute the rest of the estimate. This time their outputs will be deterministic, but there is no point in doing uncomputation since the superposition has already collapsed. Instead, we do the same thing as for the $k=0$ estimator: compute the next bit and some garbage, and discard the garbage. The final answer is either $\text{floor}(2^n \lambda_j)$ or $\text{floor}(2^n \lambda_j) - 1$ for the same reason as in Corollary~\ref{cor:improvedphaseestimation}.

\end{proof}

\section{Performance Comparison} \label{sec:performance}

Above, we have presented a modular framework for phase and energy estimation using the key ingredients in Theorem~\ref{thm:iterativephaseestimation} and Theorem~\ref{thm:iterativeenergyestimation} respectively. These results already demonstrate several advantages over textbook phase estimation as presented in Proposition~\ref{prop:phaseestimation}. 

First, they eliminate the QFT and do not require a sorting network to perform median amplification. Instead, they rely on just a single tool: singular value transformation. 

Second, they require far fewer ancillae. Our improved phase estimation algorithm requires no ancillae at all, and is merely `with phases' so arguably does not even need uncomputation for some applications. On the other hand, textbook phase estimation requires $O( (n+\log(\alpha^{-1})) \log(\delta^{-1}) )$ ancillae in order to implement median amplification.

Given a block-encoding of a Hamiltonian with $a$ ancillae, then our energy estimation algorithm requires $a+n+3$ ancillae. But in order to even compare Proposition~\ref{prop:phaseestimation} to Theorem~\ref{thm:iterativeenergyestimation} we need a method to perform energy estimation using textbook phase estimation. This is achieved through Hamiltonian simulation.

Which method of Hamiltonian simulation is best depends on the particular physical system involved. Hamiltonian simulation using the Trotter approximation can perform exceedingly well in many situations \cite{2012.09194}. However, in our analysis we must be agnostic to the particular Hamiltonian in question, and furthermore need a unified method for comparing the performance. Hamiltonian simulation via singular value transformation \cite{1707.05391,1806.01838}, lets us compare Proposition~\ref{prop:phaseestimation} and Theorem~\ref{thm:iterativeenergyestimation} on the same footing. After all, this method features the best known asymptotic performance in terms of the simulation time \cite{1912.08854} in a black-box setting.

Singular value transformation constructs an approximate block-encoding of $e^{iHt}$ with ancillae. Since $e^{iHt}$ is unitary, in the ideal case the ancillae start in the $\ket{0}$ state and are also guaranteed to be mapped back to the $\ket{0}$ state. But in the approximate case the ancillae are still a little entangled with the remaining registers, so they become an additional source of error. Certainly the ancillae cannot be re-used to perform singular value transformation again, because then the errors pile up with each use. Thus, we are in a similar situation to Lemma~\ref{lemma:garbageremove}, where we must discard some qubits and take into account the error. 

Therefore, we must do some additional work beyond the Hamiltonian simulation method presented in \cite{1806.01838}, to turn the approximate block-encoding of a unitary into a channel that approximates the unitary in diamond norm. The main trick for this proof is to consider postselection of the ancilla qubits onto the $\ket{0}$ state. Then the error splits into two parts: the error of the channel when the postselection succeeds, and the probability that the postselection fails. 

The Hamiltonian simulation method is extremely accurate, letting us obtain block-encodings with an error that decays exponentially in the query complexity to $U_H$. Thus, the error contribution of this channel is almost entirely negligible in the performance analysis. The purpose of the argument below is really to provide a Lemma analogous to Lemma~\ref{lemma:spectraltodiamond} that lets us convert error bounds in spectral norms on block-encodings to diamond norms. That way, our entire error analysis is completely formal, as it should be for a fair comparison of all algorithms involved. Notably, an $\eps$-accurate block-encoding of a unitary is not the same thing as an $\eps$-accurate implementation of that unitary: the latter implies a channel with diamond norm error $2\eps$ while the prior yields an error $4\eps$ due to the ancilla registers.

\begin{lemma} \textbf{Hamiltonian simulation.} \label{lemma:hamsim} Say $U_H$ is a block-encoding of a Hamiltonian $H$. Then, for any $t > 0$ and $\eps > 0$ there exists a quantum channel that is $\eps$-close in diamond norm to the channel induced by the unitary $e^{i H t}$. This channel can be implemented using
    \begin{align}
        3 \cdot r\left( \frac{et}{2} , \frac{\eps}{24} \right) + 3
    \end{align}
queries to controlled-$U_H$ or controlled-$U_H^\dagger$.
\end{lemma}

\begin{proof} This is an extension of Theorem~58 of \cite{1806.01838}, which states that there exists a block-encoding $U_A$ of a matrix $A$ such that $|A - e^{iHt}| \leq \eps$ with query complexity $3 r \left( \frac{et}{2} , \frac{\eps}{6}  \right) $. This result leverages the Jacobi-Anger expansion (Lemma~\ref{lemma:jacobianger}) to construct approximate block-encodings of $\sin(t H)$ and $\cos(t H)$ and uses linear combinations of unitaries to approximate $e^{iHt} / 2$. Then it uses oblivious amplitude amplification to get rid of the factor of $1/2$, obtaining $U_A$.  If $U_H$ is a block-encoding with $a$ ancillae, then $U_A$ has $a+2$ ancillae. 

    All that is left to do is to turn this block-encoding into a quantum channel that approximately implements $e^{iHt}$. To do so, we just initialize the ancillae to $\ket{0^{a+2}}$, apply $U_A$, and trace out the ancillae. We can write this channel $\Lambda$ as:
    \begin{align}
        \Lambda(\rho) :=  \sum_{i}  (\bra{i} \otimes I)  U_A ( \ket{0^{a+2}}\bra{0^{a+2}} \otimes \rho) U_A^\dagger ( \ket{i} \otimes I) 
    \end{align}
   To finish the theorem we must select an $\eps$ so that the error in diamond norm of $\Lambda$ to the channel implemented by $e^{iHt}$ is bounded by $\delta$.  To do so, we write $\Lambda$ as a sum of postselective channels $\Lambda_i(\rho)$:
    \begin{align}
        \Lambda_i(\rho) &:= (\bra{i} \otimes I)  U_A ( \ket{0}\bra{0} \otimes \rho) U_A^\dagger ( \ket{i} \otimes I) 
    \end{align}
That way $\Lambda = \sum_i \Lambda_i$. If we let $\Gamma_{e^{iHt}} := e^{iHt}\rho e^{-iHt}$, then we can bound the error in diamond norm using the triangle inequality:
    \begin{align}
        \left| \Lambda - \Gamma_{e^{iHt}}  \right|_\diamond  \leq \left| \Lambda_0 -  \Gamma_{e^{iHt}}   \right|_\diamond + \left| \sum_{i > 0}  \Lambda_i \right|_\diamond 
    \end{align}
Now we proceed to bound the two terms individually. Observe that:
    \begin{align}
        \Lambda_0(\rho) &:= (\bra{0} \otimes I)  U_A ( \ket{0}\bra{0} \otimes \rho) U_A^\dagger ( \ket{0} \otimes I)  = A \rho A^\dagger
    \end{align}
    Since $|A - e^{iHt}| \leq \eps$, we can invoke Lemma~\ref{lemma:spectraltodiamond} and see that $\left| \Lambda_0 -  \Gamma_{e^{iHt}}   \right|_\diamond  \leq  2\eps $, and we are done with the first term.

    To bound $\left| \sum_{i > 0}  \Lambda_i \right|_\diamond$, we first observe that $\sum_{i > 0 }\Lambda_i = \Lambda - \Lambda_0$. Second, we observe that the $\Lambda_i$  are all positive semi-definite, so $|\Lambda_i(\rho)|_1 = \text{Tr}(\Lambda_i(\rho) )$.  Plugging in the definition of the diamond norm, we can compute:
    \begin{align}
        \left| \sum_{i > 0}  \Lambda_i \right|_\diamond &= \sup_\rho \left| \sum_i (\Lambda_i \otimes \mathcal{I})(\rho)  \right|_1\\
        &\leq \sup_\rho \text{Tr}\left( \sum_i (\Lambda_i \otimes \mathcal{I})(\rho) \right)\\
        &= \sup_\rho \text{Tr}\left( (\Lambda \otimes \mathcal{I})(\rho) - (\Lambda_0\otimes\mathcal{I})(\rho) \right)\\
        &= 1 - \inf_\rho \text{Tr}( A\rho A^\dagger  ) 
    \end{align}
    If we let $E := e^{iHt} - A $ so that $|E| \leq \eps$, and plug into the above expression, we get:
    \begin{align}
        \text{Tr}(A\rho A^\dagger) = \text{Tr}( e^{iHt} \rho e^{-iHt} - E\rho e^{-iHt} - e^{iHt} \rho E^\dagger + E\rho E^\dagger  ) \geq 1 - 2\eps + \eps^2
    \end{align}
  Putting everything together we obtain $   \left| \Lambda - \Gamma_{e^{iHt}}  \right|_\diamond \leq  2\eps + 2\eps - \eps^2 \leq 4\eps$. So if we select $\eps := \delta/4$ we obtain the desired bound.

\end{proof}

The above proof uses the trick for the proof of the block-measurement theorem that we mentioned in the introduction. We will need it again in the proof Theorem~\ref{thm:blockmeasure}. Thus, while we are at it, we may as well state the generalization of this result to any block-encoded unitary as a proposition. 

\begin{proposition} \label{prop:blocktochannel} Say $U_A$ is a block-encoding of $A$, which is $\eps$-close in spectral norm to a unitary $V$. Then there exists a quantum channel $4\eps$-close in diamond norm to the channel $\rho \to V\rho V^\dagger$.
\end{proposition}
\begin{proof} Lemma~\ref{lemma:hamsim} proved this result with $V = e^{iHt}$. The exact same argument holds for abstract $V$.
\end{proof}

Returning to the ancilla discussion, say that the block-encoding of the input Hamiltonian has $a$ ancillae. Then we see that Theorem~\ref{thm:iterativeenergyestimation} requires $a+n+3$ ancillae, while textbook phase estimation combined with Lemma~\ref{lemma:hamsim} requires $O( a + (n + \log(\alpha^{-1}))  \log(\delta^{-1}) )$. This is because the extra ancillae required to perform the procedure in Lemma~\ref{lemma:hamsim} can be discarded and reset after each application of $e^{iHt}$. Also note that despite the fact that the above implementation of $e^{iHt}$ is not unitary, the overall protocol in Proposition~\ref{prop:phaseestimation} is still approximately invertible as required by Lemma~\ref{lemma:garbageremove}, since we can just use Lemma~\ref{lemma:hamsim} to implement $e^{-iHt}$ instead.

Next, we compare the algorithms in terms of their query complexity to the unitary $U$ or the block-encoding $U_H$. Note that this is \emph{not} the gate complexity: the number of gates from some universal gate set used to implement the algorithm. The reason for this choice is that the gate complexity of $U$ or $U_H$ depends on the application, and is likely much much larger than any of the additional gates used to implement singular value transformation. The algorithms' query complexities depend on three parameters: $\alpha, n $, and $\delta$. In the following we discuss the impact of these parameters on the complexity and show how we arrive at the 14x and 10x speedups stated in the introduction. 
\clearpage

The following jupyter notebook contains the code we used to evaluate the query complexities of the algorithms:
\begin{center}
    \url{https://github.com/qiskit-community/improved-phase-estimation/blob/main/phase_estimation_performance.ipynb}
\end{center}
The implementation of the algorithms follows the constructions in Corollary~\ref{cor:improvedphaseestimation} and Corollary~\ref{cor:improvedenergyestimation} with only one minor optimization: rather than setting $\eta_k$ based on a linear lower bound on $\cos(\pi\lambda^{(k)})$, we simply evaluate $\cos(\pi\lambda^{(k)})$ at the relevant point.

The performance in terms of $\alpha$ is plotted in Figure~\ref{fig:performance_comparison}, for fixed $n=10$ and $\delta = 10^{-30}$. This comparison shows several features. First, we see that the novel algorithms in this paper are consistently faster than the traditional methods in this regime. The only exception is Corollary~\ref{cor:improvedphaseestimation} combined with Lemma~\ref{lemma:garbageremove} to remove the phases, which is outperformed by traditional phase estimation for some value of $\alpha > 1/2$. Such enormous $\alpha$ is obviously impractical: even traditional phase estimation does not round to the nearest bit in this regime.

Second, the performance of Proposition~\ref{prop:phaseestimation} exhibits a ziz-zag behavior. This is because we reduce $\alpha$ by estimating more bits than we need and then rounding them away, a process that can only obtain $ \alpha $ of the form $2^{-1-r}$ for integer $r$. When $\alpha > \frac{1}{2}$ then we no longer use the rounding strategy since we can achieve these values with median amplification alone. Therefore, for large enough $\alpha$ the line is no longer a zig-zag and begins to be a curve with shape $\sim \alpha^{-2}$.

The zig-zag behavior makes comparison for continuous values of $\alpha$ complicated. In our analysis we would like to compare against the most efficient version of traditional phase estimation, so we continue our performance analysis where $\alpha$ is a power of two, maximizing the efficiency (but minimizing our speedup). In Figure~\ref{fig:speedup} we show the performance when vary $\alpha,n$ and $\delta$ independently.

We see that once $n \gtrsim 10$, $\alpha \lesssim 2^{-10}$ and $\delta \lesssim 10^{-30}$ the speedup stabilizes at about 14x for phase estimation and 10x for energy estimation. Our method therefore shows an significant improvement over the state of the art. 

Of course, several assumptions needed to be made in order to claim a particular multiplicative speedup. Many of these assumptions were made specifically to maximize the performance of the traditional method. For example, we count performance in terms of query complexity rather than gate complexity, which neglects all the additional processing that phase estimation needs to perform for median amplification. This method of comparison favors the traditional method, since it neglects the fact that our new methods require significantly less ancillary processing. Furthermore, we select $\alpha$ to be a power of two, so that phase estimation is maximally efficient. However, we also assume that $n$ is large enough and that $\alpha$ and $\delta$ are small enough that the speedup is stable. If the accuracy required is not so large, then the speedup is less significant.

Which method is best in reality will depend on the situation. In particular, the appropriate choice of $\alpha$ when studying real-world Hamiltonians remains an interesting direction of study. In reality it will not be possible to guarantee a rounding promise, so one's only option is to pick a small value of $\alpha$ and hope for the best. How small of an $\alpha$ is required?

\clearpage

\begin{figure}[h]
    \centering
    \includegraphics[width=0.8\textwidth]{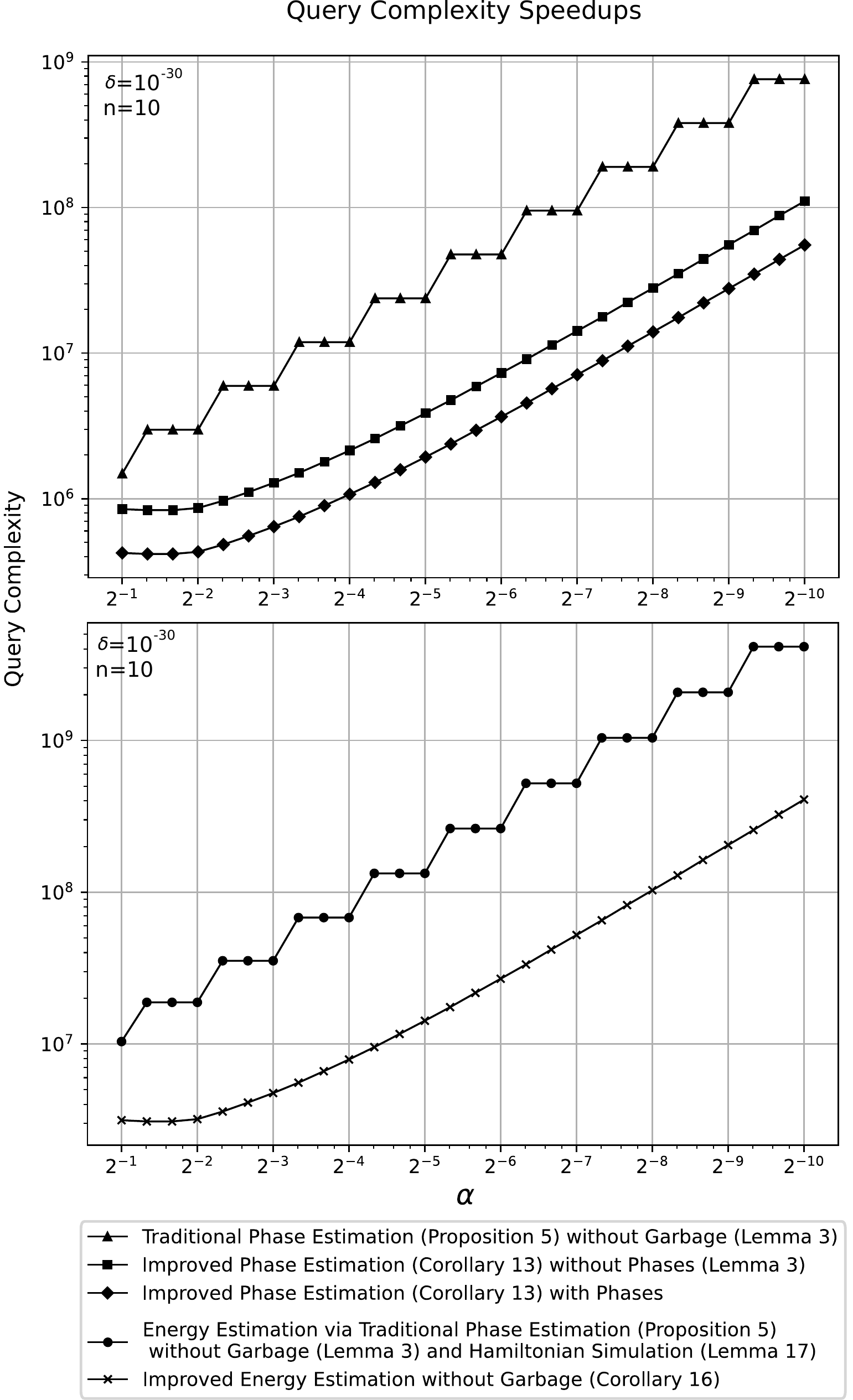}
    \caption{ \label{fig:performance_comparison}  Performance of estimation algorithms presented in this paper. Phase estimation algorithms are presented with slim lines on the left, and energy estimation algorithms are presented with thick lines on the right. The zig-zag behavior of phase estimation is explained by the method by which Proposition~\ref{prop:phaseestimation} reduces $\alpha$: by estimating more bits than needed and then rounding them. Thus, the $\alpha$ achieved by phase estimation is always a power of two. }
\end{figure}

\begin{figure}[h]
    \centering
    \includegraphics[width=0.85\textwidth]{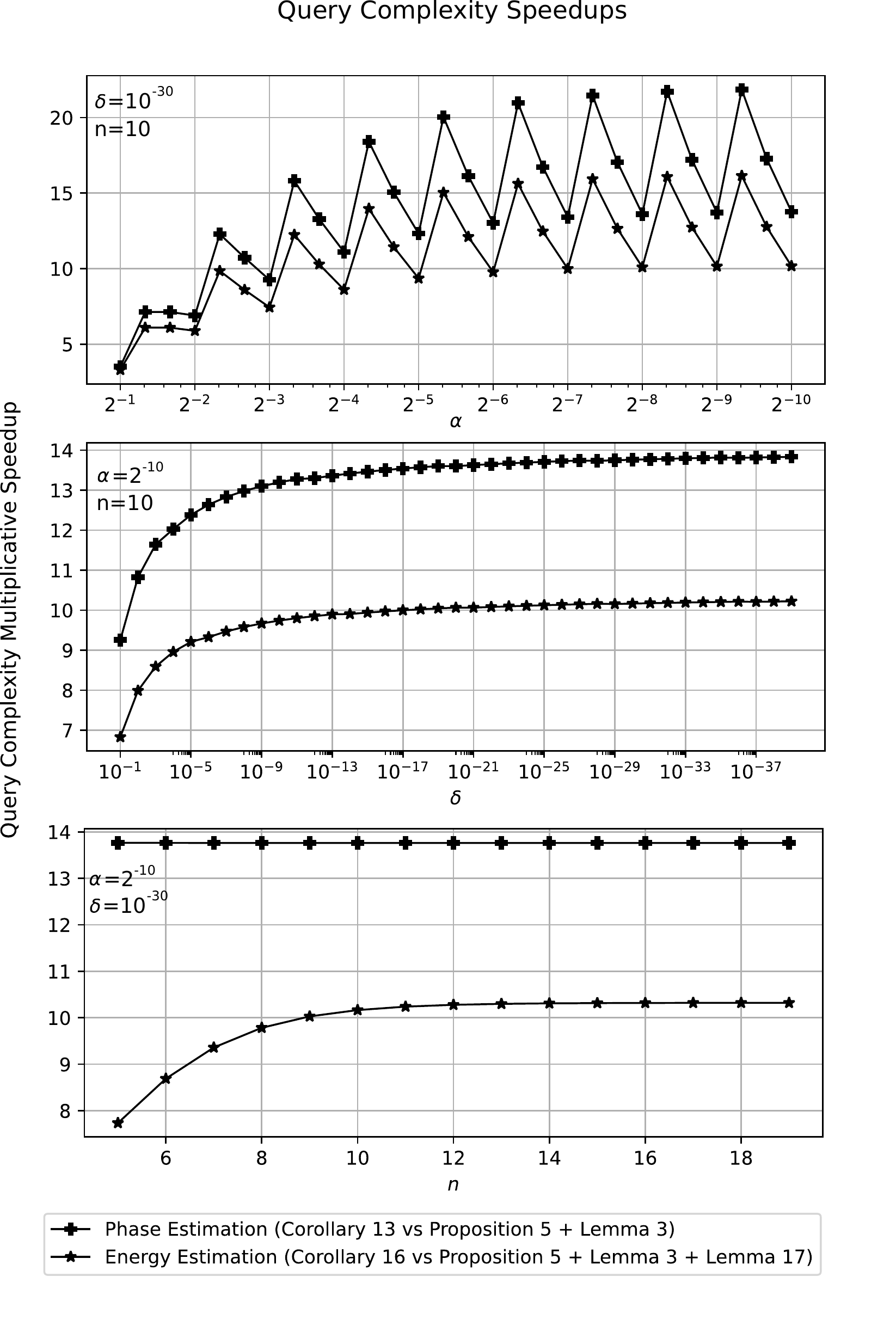}
    \caption{ \label{fig:speedup} Speedup over traditional methods by our new methods for phase and energy estimation. When $n \gtrsim 10$, $\delta \lesssim 10^{-30}$ and $\alpha \lesssim 2^{-10}$ the speedups become stable.  Together with Figure~\ref{fig:performance_comparison} we can conclude that the speedup for phase estimation is about 14x and the speedup for energy estimation is about 10x.  } 
\end{figure}

\clearpage

\section{Block-Measurement}  \label{sec:blockmeas}

We have presented improved algorithms for phase and energy estimation. In this section we prove the block-measurement theorem from the introduction. The goal of the block-measurement protocol is the following: given an approximate block-encoding of a projector $\Pi$, implement a channel close to the unitary:
\begin{align}
    \ket{1} \otimes \Pi + \ket{0} \otimes (I - \Pi)
\end{align}

When the block-encoding of $\Pi$ is exact, then the above is easily accomplished through linear combinations of unitaries and oblivious amplitude amplification. In particular, we can rewrite the above as $\ket{0}\otimes I -\sqrt{2}\ket{-}\otimes\Pi $, so we need to amplify away a factor of $1/(1+\sqrt{2})$ from the linear combination. Following Theorem~28 of \cite{1806.01838}, we observe that $T_5(x)$ is the first Chebyshev polynomial that has a solution $T_5(x) = \pm 1$ such that $x < 1/(1+\sqrt{2})$. We nudge the factor down to the solution $x$ with some extra postselection, and then we meet the conditions of this theorem. This demands five queries to the block-encoding of $\Pi$. Then we invoke our Proposition~\ref{prop:blocktochannel} to turn the block-encoding of $\ket{1} \otimes \Pi + \ket{0} \otimes (I - \Pi)$ into a channel.

However, the above strategy is both more complicated and more expensive than necessary -  we can do this in just two queries. Say $U_\Pi$ is a block-encoding of $\Pi$ with $m$ ancillae. Then, let $V_\Pi$:
    \begin{align}
        \label{eq:blockmeasure} V_\Pi  :=  \hspace{5mm} \begin{array}{c}\Qcircuit @C=1em @R=1em {
                     & \qw & \targ & \qw  & \qw\\
                 & \multigate{1}{U_\Pi} & \ctrl{-1} & \multigate{1}{U_\Pi^\dagger} & \qw \\
             & \ghost{U_A} & \qw & \ghost{U^\dagger_A} & \qw
    }\end{array} 
    \end{align}
where the CNOT above refers to $X \otimes \ket{0^m}\bra{0^m} + I \otimes (I - \ket{0^m}\bra{0^m})$. Then, $V_\Pi$ satisfies:
\begin{align}
      &(I \otimes \bra{0^m} \otimes I)  V_\Pi (\ket{0} \otimes \ket{0^m} \otimes I)\\
    = & \ket{1} \otimes \Pi + \ket{0} \otimes (I - \Pi )
\end{align}
So viewing the $\ket{0^m}$ register as the postselective part of a block-encoding, we see that $V_\Pi$ is a block-encoding of the desired operation. Then we invoke Proposition~\ref{prop:blocktochannel} to construct a channel $\Lambda_\Pi$ from $V_\Pi$. Now all that is left to do is the error analysis.

In fact, when $m=1$, then there is an extent that we do not even need the modified CNOT gate and can get away with just a single query. We used this fact in Theorem~\ref{thm:iterativephaseestimation}. This is because if $\Pi = \sum_j \alpha_j \ket{\psi_j}\bra{\psi_j}$, then we can write:
\begin{align}
    U_\Pi (\ket{0} \otimes I)   &=  \sum_j \left( \alpha_j \ket{0} + \beta_j \ket{1} \right) \ket{\psi_j}\bra{\psi_j}
\end{align}
where $\beta_j$ is some amplitude satisfying $|\beta_j|^2 + |\alpha_j|^2 = 1$. To compare, the desired operation is:
    \begin{align}
        \sum_j \left( \alpha_j \ket{1} + (1 - \alpha_j) \ket{0} \right) \ket{\psi_j}\bra{\psi_j}
    \end{align}
Since $\alpha_j \in \{0,1\}$, we know that $\beta_j = e^{i\phi_j}(1-\alpha_j)$ for some phase $\phi_j$ that may depend on $\ket{\psi_j}$. So, with the exception of the phase correction, we are already one $X$ gate away from the desired unitary.  

Is it possible to remove this phase correction without uncomputation? If $U_\Pi$ is obtained through singular value transformation of some real eigenvalues $\lambda_j$ then we actually have more control than Lemma~\ref{lemma:svt} would indicate. Looking at Theorem~3 of \cite{1806.01838}, we can actually select polynomials $P,Q$ such that $\alpha_j = P(\lambda_j)$ and $\beta_j = Q(\lambda_j)$, provided that $P,Q$ satisfy some conditions including $|P(x)|^2 + (1-x^2)|Q(x)|^2 = 1$. For what positive-real-valued polynomials $P(x)$ is it possible to choose a positive-real-valued $Q(x)$ that satisfies this relation, thus removing the phases $e^{i\phi_j}$? We leave this question for future work.

Now we proceed to the formal statement and error analysis of the block-measurement theorem.

\begin{theorem} \textbf{Block-measurement.} \label{thm:blockmeasure} Say $\Pi$ is a projector, and $A$ is a hermitian matrix satisfying $|\Pi - A^2| \leq \eps $. Also, $A$ has a block-encoding $U_A$. Then the channel $\Lambda_A$, constructed in the same way as $\Lambda_\Pi$ above just with $U_A$ instead of $U_\Pi$, approximates $\Lambda_\Pi$ in diamond norm:
    \begin{align}
        \left| \Lambda_A  - \Lambda_\Pi  \right|_\diamond \leq 4\sqrt{2}\eps 
    \end{align}
\end{theorem}
\begin{proof} Just like $V_\Pi$, $V_A$ is a block-encoding of the unitary:
    \begin{align}
        X \otimes A^2 + I \otimes (I - A^2)
    \end{align}
The distance in spectral norm to $V_\Pi =  X \otimes \Pi + I \otimes (I - \Pi)$ is:
    \begin{align}
        & \Big|   \ket{1} \otimes \Pi + \ket{0} \otimes (I - \Pi) -  \ket{1} \otimes A^2 - \ket{0} \otimes (I - A^2)   \Big|\\
       = & \Big|   \ket{1} \otimes (\Pi - A^2)  + \ket{0} \otimes (A^2 - \Pi)   \Big|\\
       = & \sqrt{2} \left|   \frac{\ket{0} - \ket{1}}{\sqrt{2}} \otimes (\Pi - A^2)   \right|\\
       = & \sqrt{2} \left|\Pi - A^2  \right| \leq \sqrt{2}\eps 
    \end{align}
    So we have a block-encoding of a matrix $\sqrt{2}\eps$-close in spectral norm to the unitary $ X \otimes \Pi + I \otimes (I - \Pi)$. Now we just use Proposition~\ref{prop:blocktochannel}, which gives a channel with diamond norm accuracy $4\sqrt{2}\eps$. Then we get the desired operation by initializing an extra qubit to $\ket{0}$ before applying the channel.
\end{proof}


What happens when $A$ is not a projector? For simplicity, say the input state to the circuit is $\ket{\Psi}$, which happens to be an eigenstate of $A$. Then:
\begin{align}
    V_A \left( \ket{0} \otimes \ket{0^m} \otimes \ket{\Psi} \right) &= (\ket{1} \otimes  U_A^\dagger\ket{0^m}\bra{0^m}  U_A \ket{0^m}+ \ket{0} \otimes  U_A^\dagger \left(I - \ket{0^m}\bra{0^m} \right) U_A \ket{0^m} ) \otimes  \ket{\Psi}\\
    &= (-\sqrt{2}\ket{-} \otimes  U_A^\dagger\ket{0^m}\bra{0^m}  U_A \ket{0^m} + \ket{0} \otimes U_A^\dagger U_A \ket{0^m} ) \otimes  \ket{\Psi}
\end{align}

We are interested in the state obtained by tracing out the $\ket{0^m}$ register above. For $j \in \{0,1\}^m$, we let $\gamma_j$ abbreviate a matrix element of $U_A$:
\begin{align}
    \gamma_j \ket{\Psi} &:=  (\bra{0^m} \otimes I) U_A \left( \ket{j} \otimes \ket{\Psi} \right)\\
    \left( I \otimes \bra{j} \otimes I \right)  V_A \left( \ket{0} \otimes \ket{0^m} \otimes \ket{\Psi} \right) &= \left( -\gamma_j^\dagger \gamma_0\sqrt{2} \ket{-}  + \delta_{j,0} \ket{0} \right) \otimes \ket{\Psi}
\end{align}
Note that $\gamma_0$ is the eigenvalue of $A$ corresponding to $\ket{\Psi}$. Then, the reduced density matrix after tracing out the middle register is:
\begin{align}
    & \sum_{j \in \{0,1\}^m}   \left(-  \gamma_j^\dagger \gamma_0\sqrt{2} \ket{-}  + \delta_{j,0} \ket{0} \right)  \left( - \gamma_j \gamma_0^\dagger \sqrt{2} \bra{-}  + \delta_{j,0} \bra{0} \right) \otimes \ket{\Psi}\bra{\Psi}\\
    =&  \sum_{j \in \{0,1\}^m}   \left( 2 |\gamma_j|^2 |\gamma_0|^2 \ket{-}\bra{-}  - \sqrt{2} \delta_{j,0} \gamma_j \gamma_0^\dagger \ket{0}\bra{-} - \sqrt{2} \delta_{j,0} \gamma_j^\dagger \gamma_0\ket{-}\bra{0} +  \delta_{j,0} \ket{0}\bra{0} \right) \otimes \ket{\Psi}\bra{\Psi}\\
    =&   \left( 2 |\gamma_0|^2 \ket{-}\bra{-}  - \sqrt{2 }|\gamma_0|^2 \ket{0}\bra{-} - \sqrt{2} |\gamma_0|^2\ket{-}\bra{0} +  \ket{0}\bra{0} \right) \otimes \ket{\Psi}\bra{\Psi}\\
    =&   \left(  |\gamma_0|^2 \ket{1}\bra{1}  + (1-|\gamma_0|^2) \ket{0}\bra{0} \right) \otimes \ket{\Psi}\bra{\Psi}
\end{align}
This somewhat complicated calculation produced a simple result: tracing out the $\ket{0^m}$ register collapses the superposition on the output register. A more general version of this calculation where the input state is not an eigenstate of $A$ demonstrates that this collapse also damages the superposition on the input register (although it does not fully collapse it). Intuitively, we can see why this is the case even without doing the full calculation: since measuring the middle register collapses the superposition that encodes the eigenvalues, the environment that traced out the middle register thus learns information about the eigenvalues. But the distribution over eigenvalues also contains information about the input superposition over eigenvectors, so therefore the input state will be damaged.

In essence, we have re-derived the fact that uncomputation is impossible unless the output of the computation is deterministic.

\section{Non-Destructive Amplitude Estimation} \label{sec:ampest}

In this section we show how to use our energy estimation algorithm to perform amplitude estimation. The algorithms for energy and phase estimation demanded a rounding promise on the input Hamiltonian, which guarantees that they do not damage the input state even if it is not an eigenstate of the unitary or Hamiltonian of interest. However, as we shall see, this is not a concern for amplitude estimation. This is because we can construct a Hamiltonian whose eigenstate is the input state.

Say $\Pi$ is some projector and $\ket{\Psi}$ is a quantum state. Non-destructive amplitude estimation obtains an estimate of $a = |\Pi\ket{\Psi}|$ given exactly one copy of $\ket{\Psi}$, and leaves that copy of $\ket{\Psi}$ intact. This subroutine is explicitly required by the algorithms in \cite{1907.09965, 2009.11270}, which can be used to perform Bayesian inference, thermal state preparation and partition function estimation. However, it is also quite useful in general when $\ket{\Psi}$ is expensive to prepare. For example, when estimating the expectation of an observable on the ground state of a Hamiltonian, preparing the ground state can be very expensive. Thus, it may be practical to only prepare it once.

The only previously known algorithm for non-destructive amplitude estimation is given in \cite{1907.09965}. It works via several invocations of amplitude estimation according to \cite{0005055}, which is based on phase estimation. We argue that our new algorithm has several performance advantages over the prior art. However, these advantages are not very quantifiable, preventing us from computing a numerical constant-factor speedup. The advantages are as follows:

\begin{itemize}
    \item Our new algorithm requires dramatically fewer ancillae. This is because \cite{1907.09965} relies on phase estimation with median amplification. As argued above, median amplification requires $O(n \log(\delta^{-1}))$ ancillae. In contrast, we simply use the protocol for energy estimation which requires $n + O(1)$ ancillae.
    \item Our new algorithm runs in a fixed amount of time: one application of our energy estimation algorithm suffices. In contrast, \cite{1907.09965}'s algorithm works by repeatedly attempting to `repair' the input state. This process succeeds only with probability $1/2$, so while the expected number of attempts is constant, the resulting algorithm is highly adaptive and may need a variable amount of time.
    \item Our new algorithm does not require knowledge of a lower bound on the amplitude $a$. The `repair' step in \cite{1907.09965} itself involves another application of phase estimation, which must produce an estimate with enough accuracy to distinguish $\arcsin(a)$ and $-\arcsin(a)$.  This also implies that \cite{1907.09965} can only produce relative-error estimates, even when an additive-error estimate might be sufficient, as it is in \cite{2009.11270}.
    \item While, due to the above differences, it is not really possible to perform a side-by-side constant-factor comparison of the algorithms, we do expect to inherit a modest constant-factor speedup from our energy estimation algorithm. \cite{0005055} has no need to perform Hamiltonian simulation, so we expect the speedup to look more like the one in the case of phase estimation. Since no rounding promise is required, making $\alpha$ tiny is not really necessary. But it \emph{is} necessary that $\alpha < 1/2$ since this makes traditional phase estimation round correctly. Looking at Figure~\ref{fig:performance_comparison}, we already see constant factor speedups in this regime.
\end{itemize}

Another minor advantage of our method over \cite{1907.09965} is that it estimates $a^2$ directly, rather than going through the Grover angle $\theta := \arcsin(a)$. Coherently evaluating a sine function in superposition to correct this issue, while possible, will have an enormous ancilla overhead. $a^2$ is the probability with which a measurement of $\ket{\Psi}$ yields a state in $\Pi$, which may be useful directly. 

We note that an additive error estimate of $a$ is slightly stronger than an additive error estimate of $a^2$. Moreover this accuracy seems to be necessary for some versions of quantum mean estimation: see for example Appendix~C of \cite{2009.11270}. If the amplitude $a$ is desired rather than $a^2$, then, since the algorithm proceeds by obtaining a block-encoding of $a^2$, one can use singular value transformation to make a block-encoding of $a$ instead. A polynomial approximation of $\sqrt{x}$ could be constructed from Corollary~66 of \cite{1806.01838}. We leave a careful error analysis of a direct estimate of $a$ to future work.

\begin{corollary} \label{cor:ampest} \textbf{Non-destructive amplitude estimation.} Say $\Pi$ is a projector and $R_\Pi$ is a unitary that reflects about this projector:
    \begin{align}
        R_\Pi := 2\Pi - I
    \end{align}
    Let $\ket{\Psi}$ be some quantum state such that $a := |\Pi  \ket{\Psi}|$. Let $M := \text{floor}(a^2 2^{n})$. Then for any positive integer $n$ and any $\delta > 0$ there exists a quantum channel that implements $\delta$-approximately in diamond norm the map:
    \begin{align}
        \ket{0^n}\bra{0^n} \otimes \ket{ \Psi }\bra{\Psi} \to \left(p \ket{M}\bra{M} + (1-p) \ket{M-1 \text{ mod }2^n}\bra{M-1 \text{ mod }2^n}   \right)  \otimes  \ket{\Psi}\bra{\Psi}
    \end{align}
    for some probability $p$ (i.e., there is some probability that instead of obtaining $\text{floor}(a^2 2^{n})$ we obtain $\text{floor}(a^2 2^{n}) - 1$ ). This channel uses $O\left(2^n \log(\delta^{-1}) \right)$ controlled applications of $R_\Pi$ and $R_{\ket{\Psi}} :=  2\ket{\Psi}\bra{\Psi} - I$.
\end{corollary}

\begin{proof} We use linear combinations of unitaries to make block-encodings of $\Pi$ and $\ket{\Psi}\bra{\Psi}$.     

    \begin{align}
        \Pi = \frac{I + R_\Pi}{2}, \hspace{1cm} \ket{\Psi}\bra{\Psi} = \frac{I + R_{\ket{\Psi}}}{2}
        \end{align}
        Then, we multiply these projectors together to make a block-encoding of $A$:
    \begin{align}
        A :=  \ket{\Psi}\bra{\Psi} \cdot \Pi \cdot \ket{\Psi}\bra{\Psi} = a^2 \ket{\Psi}\bra{\Psi}
    \end{align}
    Now we simply invoke the algorithm described in Corollary~\ref{cor:improvedenergyestimation} for the case when no rounding promise is present, and apply it to the input state $\ket{\Psi}$. Since $\ket{\Psi}$ is an eigenstate of $A$, we are guaranteed to measure either $\text{floor}(a^2 2^{n-1})$ or $\text{floor}(a^2 2^{n-1}) - 1$ mod $2^n$.
\end{proof}

Observe that the random error only occurs for particular values of $a^2$, which is where the rounding promise is violated. We can ignore the rounding promise precisely because the input state $\ket{\Psi}$ is an eigenstate of the Hamiltonian: an incorrect estimate does not damage the input state. 


\section{Acknowledgments}

The author thanks Scott Aaronson, Andrew Tan, Pawel Wocjan, William Zeng, Yosi Atia, Andr\'as Gily\'en, Sukin Sim, Philip Jensen, Lasse Bjørn Kristensen and Al\'an Aspuru-Guzik for helpful conversations. This work was supported by Scott Aaronson's Vannevar Bush Faculty Fellowship from the US Department of Defense. The v4 updates were performed while the author was at IBM Quantum.

\end{document}